%% file: main.tex
\documentclass[prl,twocolumn,secnumarabic,balancelastpage,amsmath,amssymb,longbibliography,floatfix]{revtex4-2}

\usepackage{graphicx}
\usepackage{grffile}
\usepackage[dvipsnames]{xcolor}
\usepackage[hypertexnames=false,colorlinks,linkcolor=black,citecolor=black,urlcolor=black,filecolor=black]{hyperref}
\usepackage{amsmath, amssymb, amsthm}
\usepackage{booktabs}
\usepackage{multirow}
\usepackage[normalem]{ulem}
\usepackage{wasysym}
\usepackage[capitalize]{cleveref}
\usepackage{lipsum}
\usepackage{outlines}
\usepackage{enumerate}
\usepackage{setspace}
\usepackage{tabularx}
\usepackage{verbatim}
\usepackage{xcolor, soul}
\usepackage{physics}
\usepackage[export]{adjustbox}
\usepackage{tikz}
\usetikzlibrary{decorations.pathreplacing}
\usetikzlibrary{arrows,arrows.meta,calc,shapes.geometric,shapes.misc}
\tikzset{
  >=stealth',
  help lines/.style={dashed, thick},
  important line/.style={thick},
  connection/.style={thick, dotted},
}
\renewcommand\thesection{\arabic{section}}
\DeclareMathAlphabet{\mymathbb}{U}{BOONDOX-ds}{m}{n}
\usepackage{dsfont}
\usepackage{braket}
\usepackage{stmaryrd}
\hypersetup{
  pdfstartview={FitH},
  colorlinks=true,
  linkcolor=blue,
  citecolor=blue,
  filecolor=magenta,
  urlcolor=blue
}
\usepackage{algorithm}
\usepackage{algpseudocode}

\usepackage{wrapfig}
\newtheorem{thm}{Theorem}

\newtheorem{proposition}[thm]{Proposition}
\newtheorem{lemma}[thm]{Lemma}

\newcommand{\berninetysix}{$1.7_{-0.8}^{+1.8} \times 10^{-8}$}

\newcommand{\lerninetysix}{$2.9_{-1.5}^{+3.1} \times 10^{-11}$}
\newcommand{\nkdninetysix}{$\llbracket 1152,580,\leq 12\rrbracket$}
\newcommand{\beroneninetwo}{$1.5_{-1.0}^{+3.5} \times 10^{-10}$}

\newcommand{\leroneninetwo}{$1.3_{-0.9}^{+3.0} \times 10^{-13}$}
\newcommand{\nkdoneninetwo}{$\llbracket 2304,1156,\leq 14\rrbracket$}

\begin{document}

\title{Towards Ultra-High-Rate Quantum Error Correction\\ with Reconfigurable Atom Arrays}

\author{
  Chen Zhao$^{1,\dagger}$,
  Casey Duckering$^1$,
  Andi Gu$^2$,
  Nishad~Maskara$^{3,\dagger}$,
  and Hengyun~Zhou$^{1,4,\dagger}$
}

\affiliation{
  $^1$QuEra Computing Inc., 1380 Soldiers Field Road, Boston, MA 02135, US\\
  $^2$Department of Physics, Harvard University, Cambridge, MA 02138, USA\\
  $^3$Center for Theoretical Physics - a Leinweber Institute, Massachusetts Institute of Technology, Cambridge, MA 02139, USA\\
  $^4$Department of Electrical Engineering and Computer Science, Massachusetts Institute of Technology, Cambridge, MA 02139, USA\\
  $^\dagger$Email correspondence: czhao@quera.com, nishadma@mit.edu, hyzhou@mit.edu
}

\begin{abstract}
  Quantum error correction is widely believed to be essential for large-scale quantum computation, but the required qubit overhead remains a central challenge.
  Quantum low-density parity-check codes can substantially reduce this overhead through high-rate encodings, yet finite-size instances with practical logical error rates often achieve encoding rates only around or below $1/10$.
  Here, building on a recent ultra-high-rate construction by Kasai, we identify new structural conditions on the underlying affine permutation matrices that make encoding rates exceeding $1/2$ compatible with efficient implementation on reconfigurable neutral atom arrays.
  These conditions define a co-designed family of ultra-high-rate quantum codes that supports efficient syndrome extraction and atom rearrangement under realistic parallel control constraints.
  Using a hierarchical decoder with high accuracy and good throughput, we study the performance under a circuit-level noise model with $p=0.1\%$, achieving per-logical-per-round error rates of \leroneninetwo{} with a \nkdoneninetwo{} code and \lerninetysix{} with a \nkdninetysix{} code.
  We compare these codes against a heuristic Pareto frontier for finite-blocklength codes relating block length, encoding rate, and logical error rates, and find that our codes lie near the frontier.
  These results approach the teraquop regime, highlighting the promise of this code family for practical ultra-high-rate quantum error correction.

\end{abstract}

\maketitle

\section{Introduction}
\label{sec:intro}

Quantum error correction is essential for large-scale quantum computation, but the required qubit overhead remains a central obstacle.
For example, the celebrated surface code is projected to require several hundred to a thousand physical qubits per logical qubit~\cite{kitaev2003fault,dennis2002topological,fowler2012surface}.
Quantum low-density parity-check (qLDPC) codes provide a promising route to reducing this overhead by sharing the error correction across many logical qubits within a block~\cite{tillich2014quantum,breuckmann2021quantum,panteleev2022asymptotically,zhou2025opportunities}.
However, existing constructions that reach distances required for large-scale operations, while achieving asymptotically constant rate, often have rates around or below $1/10$~\cite{bravyi2024high,xu2024constant,tremblay2022constant,panteleev2019degenerate,breuckmann2017hyperbolic,higgott2024constructions,fahimniya2023fault,scruby2024high,lin2024quantum}.

Recent work by Kasai introduced an ultra-high-rate construction based on affine permutation matrices (APMs), achieving encoding rates near $1/2$ and large distances above 30 at block lengths around 9000~\cite{kasai2026breaking}.
This is a significant step towards the regime long familiar in classical LDPC coding~\cite{richardson2008modern}, where rates are often above $1/2$, sometimes as high as 237/256~\cite{etsi}. These codes are of great practical importance, forming the backbone of internet communication (5G, Wi-Fi, etc.) by enabling efficient and robust information transfer over noisy channels.
Similarly, quantum codes with rate $\geq 1/2$ will likely be a key ingredient in a low-overhead fault-tolerant architecture.
However, given that code properties generally scale with block length, it was unclear whether Kasai's APM-based paradigm could yield competitive performance with current experimental systems of around 1000 qubits~\cite{pause2024supercharged,manetsch2024tweezer,norcia2024iterative,ibm2023condor}, under more complex circuit-level noise models and realistic hardware layouts.

In this work, we show that ultra-high-rate quantum error correction can achieve competitive performance in experimentally-relevant regimes.
Starting from Kasai’s affine-permutation framework, we identify new structural conditions that preserve the code performance while admitting efficient implementation on reconfigurable neutral atom arrays~\cite{bluvstein2022quantum,bluvstein2024logical}.
This yields code instances with encoding rates exceeding $1/2$ at block sizes around $1000$ qubits, together with shallow syndrome extraction schedules based on simple orbit-wise shifts and small inter-orbit permutations.
To access the very low logical error rates relevant for long-lived memories, we further develop a hierarchical decoder that combines the throughput of belief propagation with an exact integer-programming fallback~\cite{delfosse2020hierarchical,poulin2008on,muller2025improved}.
Under a circuit-level noise model at $p=10^{-3}$ without idling noise, we directly observe average per-logical-per-round error rates of \lerninetysix{} and \leroneninetwo{} for a \nkdninetysix{} code and a \nkdoneninetwo{} code, respectively, bringing this family close to the teraquop regime.
We further derive heuristic finite-blocklength estimates of the tradeoff between block length and encoding rate at a given logical error rate, finding that our code constructions lie close to the Pareto frontier.
Taken together, these results show that ultra-high-rate quantum error-correcting codes can deliver substantial space-overhead reductions in practically relevant regimes, paving the way toward more efficient fault-tolerant quantum computers.

\section{Hardware co-designed code construction}
\label{sec:construction}
Reconfigurable neutral-atom arrays are a promising platform for fault-tolerant quantum computation because they combine long coherence times, high-fidelity control, and programmable non-local connectivity through coherent atom transport~\cite{bluvstein2022quantum,bluvstein2024logical,graham2022multi1,jenkins2022ytterbium,manetsch2024tweezer,zhang2025leveraging,sales2025experimental}. This flexible connectivity makes them natural candidates for implementing quantum LDPC codes, which leverage non-local circuits to reduce the space (qubit) overhead of quantum error correction.
At the same time, the native transport primitives are not fully arbitrary: common optical control tools such as crossed acousto-optic deflectors (AODs) naturally support parallel row- and column-wise motion with a product structure. The central question is therefore not only whether a code family has good parameters, but whether the corresponding reconfiguration patterns can be efficiently compiled into the structured movement primitives available in hardware~\cite{xu2024constant,viszlai2023matching}.

To achieve high encoding rate, a particularly attractive starting point is the recent ultra-high-rate Calderbank-Shor-Steane (CSS) construction introduced by Kasai~\cite{kasai2026breaking}, which showed that one can achieve encoding rates near $1/2$, at block lengths such as $\llbracket n=9216,k=4612,d\leq 32\rrbracket$, where $n$, $k$, and $d$ are the number of data qubits, the number of logical qubits, and the code distance.
The CSS code is described by $X$- and $Z$-stabilizer check matrices $H_X$, $H_Z$, which are each built from $P\times P$ permutation-matrix blocks $F_i$, $G_i$ arranged in a block-circulant pattern (Fig.~\ref{fig:pipeline} and Appendix~\ref{app:code_construction}, Eqs.~(\ref{eq:kasai_x})--(\ref{eq:kasai_z})).
The key innovation is that these blocks are \emph{affine} permutations,
\[
  x \mapsto ax+b \pmod P
\]
with $\gcd(a,P)=1$,
rather than purely circulant shifts $x \mapsto x+b \pmod P$.
Because affine permutations form a non-Abelian (non-commuting) group, they enlarge the design space and can improve girth and distance while retaining a high encoding rate.

To measure stabilizers, we use 12 blocks of size $P$ to encode the data qubits and 6 blocks of size $P$ to encode $X$ and $Z$ ancilla qubits.
As illustrated in Fig.~\ref{fig:atom_move}(a), each stabilizer is measured by a sequence of transversal gates between an ancilla block and successive data blocks, with the ancilla qubits reordered between each step.
The block-circulant structure of $H_X$ and $H_Z$ ensures that the same transition permutation is shared across multiple blocks, so all three rows of stabilizers can be measured in parallel.
The rearrangement from step $i$ to step $j$ requires applying the transition permutation $T_{ij}=F_jF_i^{-1}$ (and similarly $F_jG_i^{-1}$, $G_jG_i^{-1}$ for the $G_i$ blocks).
In general, compiling an arbitrary affine permutation into the parallel row- and column-wise moves supported by AODs requires $O(\log P)$ transport steps (Appendix~\ref{app:log_move}), which would substantially increase the syndrome extraction time.
We perform simultaneous syndrome extraction of both bases, with the schedule of the two bases staggered as in the left-right circuit approach of Refs.~\cite{strikis2026high,xu2024constant,menon2025magic,bravyi2024high}.
The syndrome extraction circuit can be further reduced to CNOT depth 7 by employing Bell pairs~\cite{cross2024improved}.

To reduce the movement cost, we co-design the code so that every transition permutation can be compiled in a small number of moves.
The main insight is that the affine
permutations can be reorganized by orbits, such that the resulting movements are simple cyclic
shifts plus inter-cycle permutations.
The first step is to choose a reference affine permutation $A$ with large orbits as a design input---that is, we first decide on a set of simple movement patterns we want the hardware to support, and then construct the code to be compatible with those patterns.
Given $A$, we examine its orbit decomposition, where one orbit is
\[
  x \rightarrow Ax \rightarrow A^2x \rightarrow \cdots \rightarrow x.
\]
If the qubits are reordered according to these orbits, the action of $A$ becomes a simple cyclic shift within each orbit (Fig.~\ref{fig:atom_move}(b)).
We then require that all transition APMs $T_{ij}$ commute with $A$.
Because commuting maps share a common eigenbasis, each $T_{ij}$ in the reordered basis is guaranteed to decompose into a simple cyclic shift within each orbit of $A$ together with permutations between orbits (Fig.~\ref{fig:atom_move}(b)).
Thus, instead of constraining $F_i$/$G_i$ directly, we design the code by requiring the transition APMs to commute with $A$.
In some settings, such as the $P=192$ example below, we find that there are multiple independent reference APMs, which serve to synchronize the relative shifts among the individual orbits.
Consequently, each step only involves around $2-3$ moves, compared with $O(\log P)$ for an unconstrained affine permutation.
We explain these design considerations in more detail in Appendix~\ref{app:apm_compile}.

\begin{figure}
  \centering
  \includegraphics[width=\linewidth]{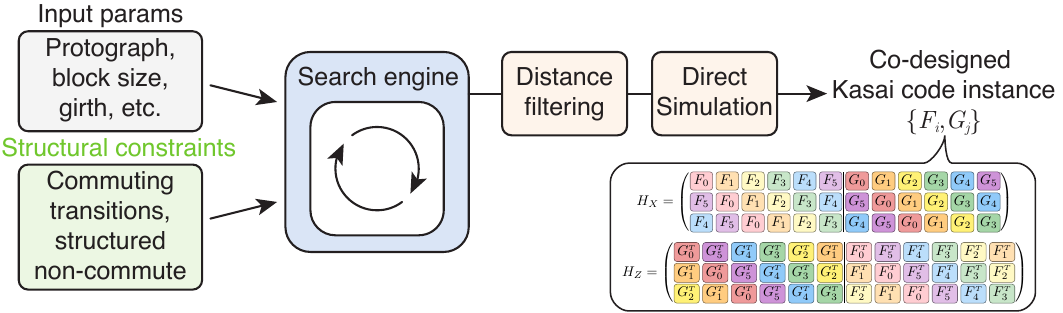}
  \caption{Illustration of code construction pipeline. We incorporate input parameters, together with structural constraints that improve code performance or simplify implementation, and search for candidates with competitive distance and performance. Inset: check matrices for Kasai's construction, color-coded to show the block-circulant structure.}
  \label{fig:pipeline}
\end{figure}

Applying these principles, we search for finite-size instances that satisfy several requirements (Fig.~\ref{fig:pipeline}): high encoding rate, guaranteed by construction; structured transition permutations based on the commuting-orbit description above; girth at least 6, to suppress short loops in the spatial decoding graph; and good finite-size distance, which we verify using open source distance-bounding tools~\cite{pryadko2022qdistrnd}.
Notably, unlike Kasai's original construction, we do not require girth-8, which imposes much stronger restrictions on the code and can be hard to achieve for small instances around $n\sim 1000$.
Empirically, we find that at moderate physical error rates, girth-6 and girth-8 instances achieve comparable performance under BP decoding with post-processing.
While constraining the non-commuting structure can in principle reduce distance, we nevertheless identify code instances with strong finite-size performance.

\begin{figure}
  \centering
  \includegraphics[width=\linewidth]{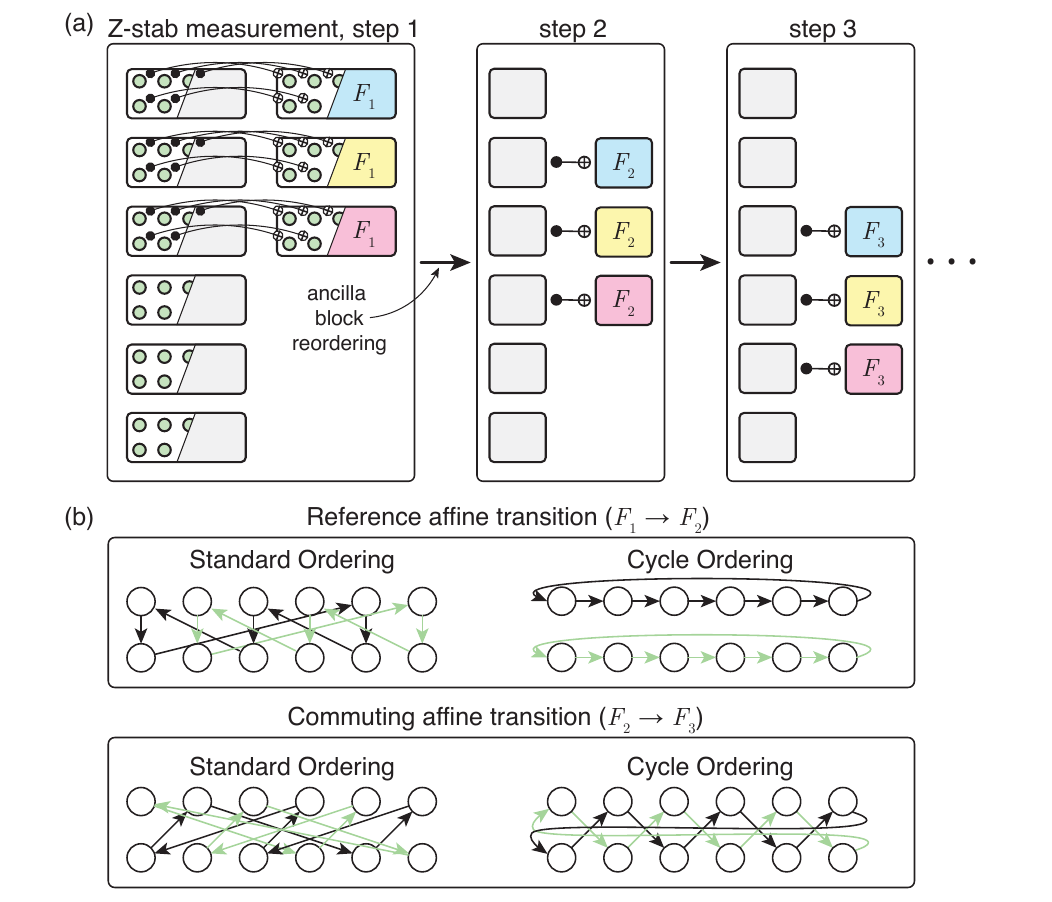}
  \caption{Illustration of atom rearrangement during syndrome extraction. (a) Each stabilizer can be measured by a sequence of transversal gates between an ancilla block of $P$ qubits, and each data block of $P$ qubits, by ordering the ancilla based on the affine permutations $F_1, F_2, ...$.
    In between transversal gates, the ancilla qubits (or data qubits) are permuted to match the target ordering at the next step.
    The block circulant structure allows all three rows of stabilizers to be measured in parallel, with identical ordering and permutations between steps.
    (b) In general, the structure of an affine permutation is complex, which generates compilation overhead when decomposed into parallel moves supported by neutral-atoms with crossed AODs. However, each APM can be decomposed into its cycle basis, and qubits can be ordered accordingly.
    Crucially, by appropriately co-designing the code, we can guarantee that all transitions in the measurement procedure commute with the \textit{reference}, and therefore map to simple column shifts and row permutations.
    Here, the reference is taken to be the first transition $F_1 \rightarrow F_2$, but this is not necessary.
  }
  \label{fig:atom_move}
\end{figure}

We focus on block lengths around several hundred to a thousand logical qubits, relevant to many proposed algorithms for ``utility-scale'' quantum computing~\cite{gidney2019how,reiher2017elucidating,lee2021even}.
We construct a \nkdninetysix{} code with $P=96$, in which all transition APMs commute with a reference APM with 3 length-32 orbits.
This results in a $3\times 32$ layout with at most one horizontal shift and one vertical shift or swap for each step.
We construct a \nkdoneninetwo{} code with length-32 orbits, as well as movement-compatible instances with $\llbracket 2304,1156,\leq 16\rrbracket$ and $\llbracket 4608,2308,\leq 22\rrbracket$. The same methodology can also be applied to larger code instances or other encoding rates.
Further details and additional examples are discussed in Appendix~\ref{app:code_construction}.
Using demonstrated experimental parameters for atom movement speeds~\cite{bluvstein2022quantum,zhou2025resource}, we estimate that one round of syndrome extraction can be completed in around $13.3\,\mathrm{ms}$ ($16.9\,\mathrm{ms}$) for the $P=96$ ($P=192$) instances, using two pairs of AODs, or around $8.3\,\mathrm{ms}$ ($9.9\,\mathrm{ms}$) using four pairs of AODs.
We estimate that the use of Bell checks to reduce syndrome extraction depth can reduce the movement time by around 25\% when using four pairs of AODs, and the use of improved movement schedules such as shortcuts to adiabaticity~\cite{hwang2025fast,cicali2025fast} have the potential to reduce movement times by $2-3\times$, bringing the syndrome extraction cycle to the $2-4$ ms range.
As we choose a long layout with horizontal dimension 32 for both code instances, the rearrangement time is comparable between the two codes.
We emphasize that while we have identified some structures that facilitate the co-design of ultra-high-rate quantum codes with hardware implementation, further improvements of the design criteria are likely possible.

\section{Noise simulations and hierarchical decoding}
\label{sec:decoding}

We now assess the performance of the constructed code instances under both phenomenological and circuit-level noise models, with a particular focus on physical error rates around 0.1\% that are common hardware targets.
The large number of logical qubits increases the entropic contribution to the logical error rate, resulting in a steeper logical error scaling near the threshold~\cite{gu2026scalable,komoto2024quantum,panteleev2019degenerate}, so we choose to probe the target error regime via direct simulation rather than extrapolation~\cite{bravyi2013simulation,beverland2025fail}.
Because the logical error rates of interest are extremely small, direct Monte Carlo simulation is only practical if the decoder is both accurate and fast.
To access this regime, we employ a hierarchical decoding pipeline~\cite{delfosse2020hierarchical} that combines the throughput of belief propagation~\cite{poulin2008on} (BP) with the accuracy of most-likely-error (MLE) decoding~\cite{landahl2011fault,cain2024correlated,beni2025tesseract} on a small subset of hard instances.

Our decoder proceeds in three stages.
We first apply a fast BP decoder (tier 1, $T_1$).
If BP does not converge, we invoke a relay-BP decoder~\cite{muller2025improved} (tier 2, $T_2$), and only if this also does not converge do we fall back to an integer-programming decoder (tier 3, $T_3$).
Empirically, we find that most decoder failures, especially in the circuit-level setting, manifest as non-convergence rather than as converged outputs in the wrong logical class.
This provides a reliable trigger for escalating to the next decoding stage.
As a result, the vast majority of shots are handled by the first two stages, while only a small fraction require the exact decoder (Fig.~\ref{fig:simulation}(a)).
In practice, this yields highly accurate decoding performance with a runtime only modestly above that of BP alone.
Additional details and timing estimates are given in Appendix~\ref{app:hierarchical_decode}.

\begin{figure}
  \centering
  \includegraphics[width=\linewidth]{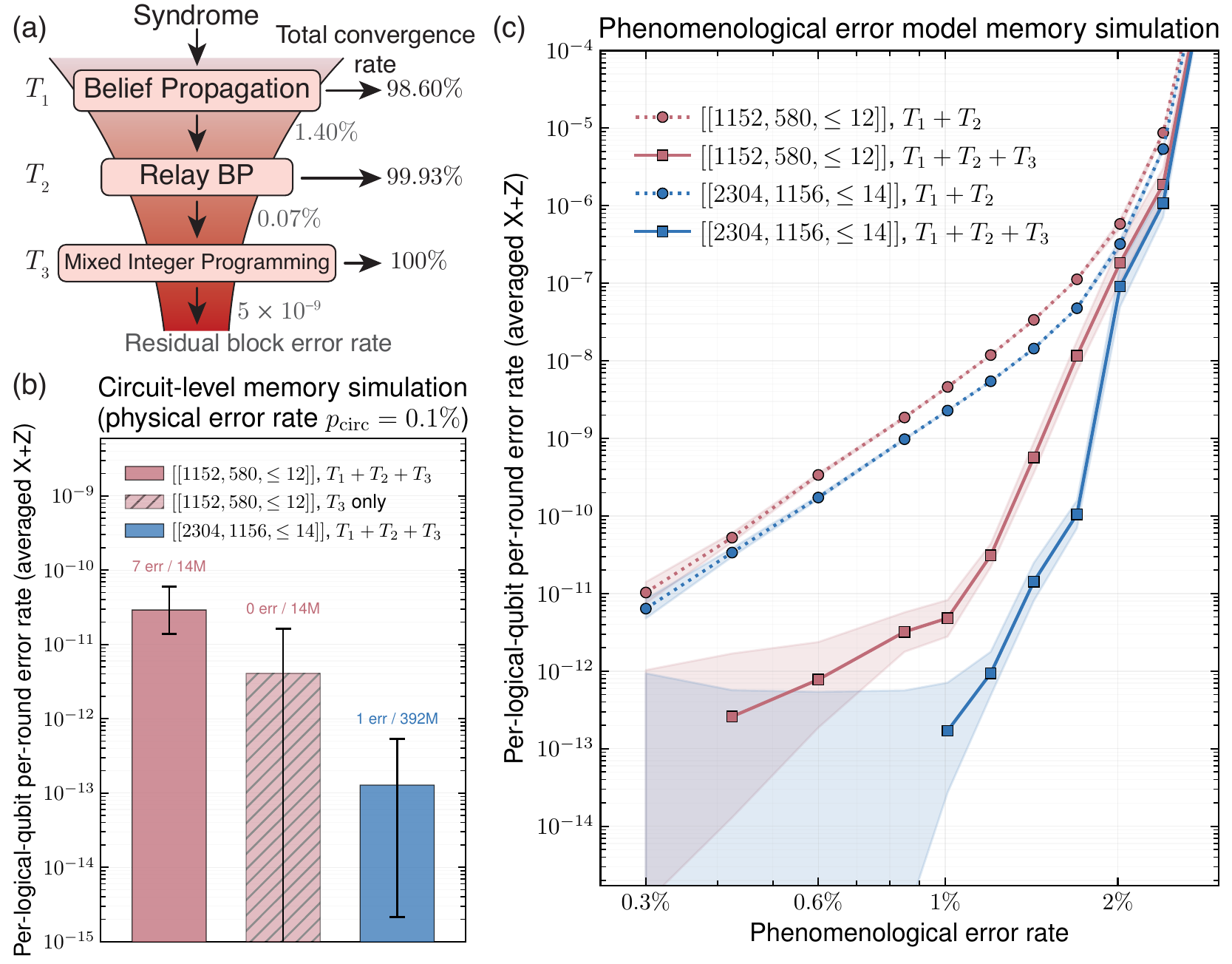}
  \caption{Decoder and simulation results. (a) Hierarchical decoder architecture with three tiers, showing the convergence rate and residual block error rate at each tier from circuit-level simulations of the \nkdoneninetwo{} code. (b) Circuit-level performance of quantum memory with $r=32$ rounds of \nkdninetysix{} and \nkdoneninetwo{} codes at $p_{\mathrm{circ}}=0.1\%$. In the ``$T_3$ only'' case, the logical error rate is estimated by redecoding all failures from earlier stages $T_1 + T_2$ using the $T_3$ decoder, where we confirm that these errors are correctable with the $T_3$ decoder. (c) Phenomenological error simulation results of quantum memory with $r=32$ rounds with various error rates. The shaded regions indicate $95\%$ Clopper-Pearson confidence intervals, while those without data points indicate cases where we were unable to observe any logical errors.}
  \label{fig:simulation}
\end{figure}

We first study a phenomenological noise model including data qubit errors and ancilla measurement errors, in order to map out the logical error scaling over a broad range of physical error rates.
More details of the error model are given in Appendix~\ref{app:simulation}.
We perform memory simulations in the \(X\) and \(Z\) basis for the \nkdninetysix{} and \nkdoneninetwo{} codes, and show the average logical error rate as a function of physical error rate in Fig.~\ref{fig:simulation}(c).
For each code, we perform decoding with 32 syndrome rounds, and we show the results with ($T_1+T_2+T_3$) and without ($T_1+T_2$) the final integer programming decoder stage.
The fallback decoder substantially improves performance in the low-\(p\) regime, and the curves exhibit a rapid initial decrease in logical error rate, consistent with the waterfall region due to the large number of logical operators for high-rate codes~\cite{panteleev2019degenerate,komoto2024quantum,gu2026scalable}.
At phenomenological error rates below $1\%$, the combined decoder shows a slower decay, limited by residual failures from the first two stages.
These phenomenological simulations provide a useful proxy for comparing candidate code instances and identifying the operating regime relevant to deep logical circuits.

We then directly test the constructed codes under a circuit-level noise model at the commonly assumed physical error rate \(p=0.1\%\).
Motivated by the long coherence times of reconfigurable neutral atom and trapped ion systems, and by the parallelizable structure of the syndrome extraction circuit, we neglect idling errors in these simulations.
For simplicity, we extract \(X\)- and \(Z\)-type syndromes sequentially in our simulations, with the ordering within each basis given by a coloration circuit~\cite{tremblay2022constant}.
On the \nkdninetysix{} code, we observe 7 logical errors out of 13.5 million shots, where all errors occur as converged errors in earlier stages but are correctable with the integer programming decoder.
We include these as errors in our accounting, as these failures will not be heralded, but note that future improvements to early stage decoders may be able to more accurately herald them.
This corresponds to a block error probability per round of \berninetysix{} and an average error rate of \lerninetysix{} per logical qubit per round.
On the \nkdoneninetwo{} code with 32 syndrome extraction rounds, we observe 1 logical error out of 392 million shots, corresponding to a block error probability per round of \beroneninetwo{} for the full simulation and an average error rate of \leroneninetwo{} per logical qubit per round.
The results place these codes close to the teraquop regime, approaching the fidelity required for large-scale quantum algorithms.

\section{Pareto frontier of ultra-high-rate codes}
\label{sec:pareto}

Having demonstrated strong performance at experimentally-relevant block lengths, we now examine the tradeoff between block length, encoding rate, and logical error rate.
At fixed physical error rate and target logical error rate, one expects the best achievable encoding rate to improve with block length.
Intuitively, larger blocks allow redundancy to be distributed more efficiently across many qubits, so that the overhead required to suppress logical errors to a given level is amortized more effectively.

A natural way to quantify this tradeoff is through finite-blocklength bounds~\cite{polyanskiy2010channel,tomamichel2016quantum}.
In classical coding theory, such bounds estimate the best achievable encoding rate at fixed blocklength, physical error rate, and target decoding error probability, and provide tight characterizations of achievable parameters at block lengths as few as a few hundred bits~\cite{polyanskiy2010channel}.
They are based on a finite size variance correction to the Hashing bound, and typically follow a Gaussian shape coming from the addition of many small error contributions close to the threshold.
Quantum analogues of these bounds have been developed for simple channels~\cite{tomamichel2016quantum}, but even for the standard depolarizing channel, the exact capacity is not known.
We therefore use the quantum hashing bound~\cite{bennett1996mixed}, commonly used as a proxy to judge asymptotic performance, and apply the same finite-blocklength estimation procedure heuristically to obtain guide curves for the quantum setting.
As shown in the appendix, the results scale as
\begin{align}
R^{*}(n,p,\varepsilon) \approx  R_H(p) -
  \sqrt{\frac{V(p)}{n}}\;\Phi^{-1}(1-\varepsilon),
\label{eq:finite}
\end{align}
where the depolarizing channel Hashing rate $R_H(p)=1-h(p)-p\log_2 3$ with $h(p)=-p\log_2 p-(1-p)\log_2(1-p)$ being the binary entropy, $V(p) = (1-p)\bigl(\log_2(1-p)\bigr)^2 +
p\left(\log_2\tfrac{p}{3}\right)^2 - \bigl(h(p) + p\log_2 3\bigr)^2$ being the dispersion term that captures the variance, and $\Phi(x)$ is the cumulative distribution function of a normal distribution.
It is worth noting that these curves should not be interpreted as rigorous bounds.
Rather, they provide a useful estimate of the scaling expected at moderate block length, and therefore a practical reference against which to judge our code families.

Figure~\ref{fig:pareto}(a) compares several quantum code families against these estimates.
We plot encoding rate as a function of block length for codes that either have been directly simulated (solid marker edges) or are extrapolated (dashed marker edges) to achieve per-logical-per-round error rates of $10^{-9}$ and $10^{-12}$ under circuit-level noise with $p_{\mathrm{circ}}=0.1\%$~\cite{bravyi2024high,webster2026pinnacle,fowler2012surface,cain2026shor,kasai2026breaking}.
As a reference, we overlay heuristic finite-blocklength estimates for the code-capacity depolarizing channel at $p=3\%$, chosen as a representative noise level because circuit-level thresholds are generally lower than code-capacity thresholds.
The circuit-level performance will depend on many factors, including the noise model, so these curves should be interpreted as rough estimates instead of tight bounds.
For this value of $p$, the asymptotic hashing bound gives an encoding rate of approximately $0.758$, whereas the finite-blocklength estimate is substantially tighter, yielding rates of only about $0.55$ at $n\sim 1000$ and about $0.1$ at $n\sim 100$, illustrating the importance of finite size corrections.

Two features of Fig.~\ref{fig:pareto}(a) are particularly notable.
First, the Gross code family~\cite{bravyi2024high} and the APM-based construction studied here lie closest to the estimated frontier in their respective block-length regimes, indicating that they are competitive realizations of high-rate quantum error correction.
Second, there is a noticeable gap in the regime of a few hundred physical qubits, where the curves suggest that encoding rates approaching $1/3$ should be achievable, yet no current construction is known to combine such block lengths with comparable logical performance.
This intermediate regime therefore appears to be a promising target for future code design.

To compare our codes against the finite-size benchmark within a matched error model, we additionally simulate them under code-capacity depolarizing noise.
For these simulations, we use a BP decoder with an integer-programming-based fallback, where the integer program is allowed to return its best feasible solution within a fixed runtime budget rather than being required to certify optimality.
Figure~\ref{fig:pareto}(b) shows that the numerical results follow the same qualitative trend as the finite-blocklength estimates, but with a residual gap of roughly a factor of two in the depolarizing error rate required to reach a given target logical error rate.
This discrepancy may indicate room for further improvement in the code construction or decoder, or it may reflect the need for sharper quantum-specific finite-blocklength estimates.
\newline

\begin{figure}
\centering
\includegraphics[width=\linewidth]{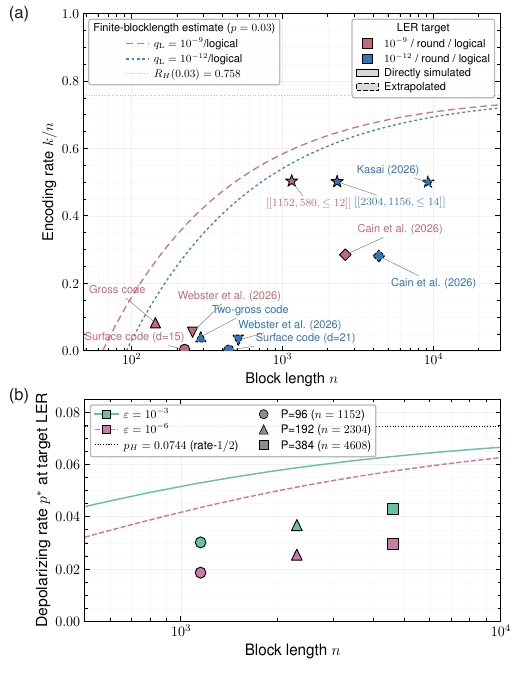}
\caption{Tradeoff between code performance and block length. (a) Encoding rate versus block length for several code instances capable of reaching per-logical-per-round error rates of $10^{-9}$ (red) or $10^{-12}$ (blue) under circuit-level noise with $p_{\mathrm{circ}}=0.1\%$. Solid markers denote codes directly verified by simulation, while dashed markers denote codes whose performance is extrapolated. (b) Depolarizing error rate required to achieve target code-capacity logical error rates of $10^{-3}$ and $10^{-6}$ as a function of block length $n$. The numerical results exhibit a modest gap relative to the corresponding finite-blocklength estimates. In both panels, the dashed curves show heuristic finite blocklength estimates based on the hashing bound for the code-capacity depolarizing channel.}
\label{fig:pareto}
\end{figure}

\section{Discussion}
\label{sec:discussion}

We have introduced a hardware-co-designed family of ultra-high-rate CSS codes that admits efficient implementation on reconfigurable atom arrays.
Combined with a hierarchical decoder, these codes achieve circuit-level logical error rates close to the teraquop regime at a physical error rate of $p=0.1\%$.
A heuristic finite-blocklength comparison further suggests that these constructions lie close to a practical Pareto frontier in their target block-length regime.
Taken together, these results show that ultra-high-rate quantum error correction can achieve competitive performance in regimes relevant to many proposed utility-scale algorithms.

There remains substantial room to improve both the code construction and the decoder.
On the coding side, a broader exploration of the affine-permutation design space, as well as combinations with other constructions such as lifted products and balanced products~\cite{panteleev2019degenerate,breuckmann2020balanced}, may reveal a more favorable Pareto frontier between rate, distance, decoding performance, and implementation cost.
Our finite-blocklength comparison suggests that the intermediate regime of a few hundred physical qubits may be a particularly promising target for such searches.
On the decoding side, the present hierarchical pipeline is designed primarily to accelerate Monte Carlo simulations, but it may also be relevant for real-time decoding in slower platforms such as neutral atoms and trapped ions.
As discussed in Appendix~\ref{app:hierarchical_decode}, the overall runtime overhead remains modest even when the fallback decoder is much slower, provided that fallback is invoked sufficiently rarely.
This suggests improving both the accuracy of the early-stage decoders and their ability to reliably herald cases requiring fallback.
Recent machine-learning decoders have also demonstrated both high speed and good accuracy, providing a promising alternative to our approach~\cite{bausch2024learning,bonilla2025neural,gu2026scalable,blue2025machine}.

The present work focuses on memory performance, and an important next step is to extend these ideas to a full architecture for fault-tolerant quantum computation~\cite{yoder2025tour,webster2026pinnacle,he2025extractors}.
Recent works have developed generalized surgery methods for Pauli-based computation on qLDPC codes~\cite{cohen2022low,cross2024improved,ide2025fault,williamson2024low,swaroop2024universal,zhang2024time}, with space overhead scaling nearly linearly in the weight of the target logical operator.
A key challenge will therefore be to identify low-weight logical bases for these codes and to construct optimized surgery gadgets.
It will also be important to explore more parallelized and time-efficient surgery schemes~\cite{zhang2024time,baspin2025fast,cowtan2025fast,xu2025batched,chang2026constant,zheng2025high}.
In such architectures, the space-time cost of magic-state production and consumption is likely to become a major bottleneck, making more efficient magic-state production and consumption especially valuable~\cite{gidney2024magic,sahay2025fold,daguerre2024code}.
More broadly, it will be interesting to understand whether the same ultra-high-rate philosophy pursued here for quantum memory can be extended to logical computation as well.
We expect that further progress in code design, logical operations, decoding, and hardware co-design will bring high-rate fault-tolerant quantum computation closer to practical realization.

\begin{acknowledgments}
  We acknowledge helpful discussions with Pablo Bonilla, Hossein Dehghani, Andrei Diaconu, Arpit Dua, Michael Gullans, Zhiyang He, Nazli Koyluoglu, Mikhail Lukin, Rohan Mehta, Varun Menon, Daniel Tan, Shengtao Wang, Qian Xu.
  This work was supported by IARPA and the Army Research Office under the Entangled Logical Qubits program (Cooperative Agreement Number W911NF-23-2-0219), the DARPA MeasQuIT program
  (HR0011-24-9-0359) and the Innovate UK grant (10179725) as part of the QNABLE project.
  NM was supported by NSF grant PHY-2325080, the U.S. Department of Energy, Office of Science, National Quantum Information Science Research Centers, Quantum Systems Accelerator (Award No. DE-SCL0000121), and by the MIT Center for Theoretical Physics (report number MIT-CTP/6027).

  Competing interests: C.Z., C.D. are employees and shareholders of QuEra Computing. H.Z. is a shareholder of QuEra Computing. A.G. and N.M. have served as consultants for QuEra Computing.

  Note added: During the preparation of this manuscript, we became aware of independent work~\cite{cain2026shor} achieving rates of $\lesssim 0.3$ with $\sim 2600$ qubits and analyzing associated architectures.
\end{acknowledgments}

\bibliography{main}

\include{appendix}

\end{document}

%% file: appendix.tex
    \appendix
\renewcommand{\thesubsection}{\thesection.\arabic{subsection}}
\makeatletter
\@addtoreset{figure}{section}
\@addtoreset{table}{section}
\@addtoreset{equation}{section}
\makeatother
\renewcommand{\thefigure}{\thesection\arabic{figure}}
\renewcommand{\thetable}{\thesection\arabic{table}}
\renewcommand{\theequation}{\thesection\arabic{equation}}
\renewcommand{\theHfigure}{appendix.\thesection.\arabic{figure}}
\renewcommand{\theHtable}{appendix.\thesection.\arabic{table}}
\renewcommand{\theHequation}{appendix.\thesection.\arabic{equation}}
\section{Details of Code Construction and Finite-Size Search Principle}
\label{app:code_construction}
This appendix supplements the code-construction discussion in Sec.~\ref{sec:construction}.
We first describe the Kasai check matrices and their parent matrix viewpoint, then explain why the truncation used to obtain a high-rate code does not introduce the most obvious low-weight logical operators, and finally summarize the finite-size search criteria used to identify the instances quoted in the main text.

\paragraph{Construction.}
The Kasai code is a CSS code~\cite{calderbank1996good,steane1996error}, defined by the check matrices of $X$ and $Z$ stabilizers,
\begin{widetext}
  \begin{align}
    H_X&=\left(
      \begin{array}{cccccc|cccccc}
        F_0 & F_1 & F_2 & F_3 & F_4 & F_5 & G_0 & G_1 & G_2 & G_3 & G_4 & G_5 \\
        F_5 & F_0 & F_1 & F_2 & F_3 & F_4 & G_5 & G_0 & G_1 & G_2 & G_3 & G_4 \\
        F_4 & F_5 & F_0 & F_1 & F_2 & F_3 & G_4 & G_5 & G_0 & G_1 & G_2 & G_3
    \end{array}\right),\label{eq:kasai_x}\\
    H_Z&=\left(
      \begin{array}{cccccc|cccccc}
        G_0^{T} & G_5^{T} & G_4^{T} & G_3^{T} & G_2^{T} & G_1^{T} & F_0^{T} & F_5^{T} & F_4^{T} & F_3^{T} & F_2^{T} & F_1^{T} \\
        G_1^{T} & G_0^{T} & G_5^{T} & G_4^{T} & G_3^{T} & G_2^{T} & F_1^{T} & F_0^{T} & F_5^{T} & F_4^{T} & F_3^{T} & F_2^{T} \\
        G_2^{T} & G_1^{T} & G_0^{T} & G_5^{T} & G_4^{T} & G_3^{T} & F_2^{T} & F_1^{T} & F_0^{T} & F_5^{T} & F_4^{T} & F_3^{T}
    \end{array}\right)\label{eq:kasai_z}.
  \end{align}
\end{widetext}
Here, $F_i$ and $G_i$ are $P\times P$ affine permutation matrices (APMs), corresponding to maps
\[
  x \mapsto ax+b \pmod P,
\]
with $\gcd(a,P)=1$ so that the map is invertible.

The displayed checks in Eqs.~(\ref{eq:kasai_x}) and (\ref{eq:kasai_z}) are obtained by retaining only the $J=3$ active block rows of the following $6\times 12$ parent matrices:
\begin{widetext}
  \begin{align}
    \hat{H}_X&=\left(
      \begin{array}{cccccc|cccccc}
        F_0 & F_1 & F_2 & F_3 & F_4 & F_5 & G_0 & G_1 & G_2 & G_3 & G_4 & G_5 \\
        F_5 & F_0 & F_1 & F_2 & F_3 & F_4 & G_5 & G_0 & G_1 & G_2 & G_3 & G_4 \\
        F_4 & F_5 & F_0 & F_1 & F_2 & F_3 & G_4 & G_5 & G_0 & G_1 & G_2 & G_3 \\
        F_3 & F_4 & F_5 & F_0 & F_1 & F_2 & G_3 & G_4 & G_5 & G_0 & G_1 & G_2 \\
        F_2 & F_3 & F_4 & F_5 & F_0 & F_1 & G_2 & G_3 & G_4 & G_5 & G_0 & G_1 \\
        F_1 & F_2 & F_3 & F_4 & F_5 & F_0 & G_1 & G_2 & G_3 & G_4 & G_5 & G_0
    \end{array}\right),\label{eq:kasai_parent_x}\\
    \hat{H}_Z&=\left(
      \begin{array}{cccccc|cccccc}
        G_0^{T} & G_5^{T} & G_4^{T} & G_3^{T} & G_2^{T} & G_1^{T} & F_0^{T} & F_5^{T} & F_4^{T} & F_3^{T} & F_2^{T} & F_1^{T} \\
        G_1^{T} & G_0^{T} & G_5^{T} & G_4^{T} & G_3^{T} & G_2^{T} & F_1^{T} & F_0^{T} & F_5^{T} & F_4^{T} & F_3^{T} & F_2^{T} \\
        G_2^{T} & G_1^{T} & G_0^{T} & G_5^{T} & G_4^{T} & G_3^{T} & F_2^{T} & F_1^{T} & F_0^{T} & F_5^{T} & F_4^{T} & F_3^{T} \\
        G_3^{T} & G_2^{T} & G_1^{T} & G_0^{T} & G_5^{T} & G_4^{T} & F_3^{T} & F_2^{T} & F_1^{T} & F_0^{T} & F_5^{T} & F_4^{T} \\
        G_4^{T} & G_3^{T} & G_2^{T} & G_1^{T} & G_0^{T} & G_5^{T} & F_4^{T} & F_3^{T} & F_2^{T} & F_1^{T} & F_0^{T} & F_5^{T} \\
        G_5^{T} & G_4^{T} & G_3^{T} & G_2^{T} & G_1^{T} & G_0^{T} & F_5^{T} & F_4^{T} & F_3^{T} & F_2^{T} & F_1^{T} & F_0^{T}
    \end{array}\right)\label{eq:kasai_parent_z}.
  \end{align}
\end{widetext}
If one were to delete rows from a generic stabilizer parity-check matrix without further structure, then the removed $X$ and $Z$ checks would naturally become candidate logical $X$ and $Z$ operators, so the distance could collapse to the weight of a deleted stabilizer.
Kasai's construction avoids this trivial failure mode by arranging the omitted rows so that they necessarily involve APM blocks that do not commute in the required way with the retained checks of the opposite basis.
As a result, the deleted rows are not themselves valid low-weight logical operators, and the high encoding rate from truncation can be retained without forcing an immediate low-distance obstruction.

\paragraph{Finite-size search.}
Our finite-size search starts from this truncated template but adds the implementation constraints described in Sec.~\ref{sec:construction}.
We start by choosing a reference APM $A$ with large orbits.
For example, in the $P=96$ case, the reference APM is composed of three orbits of size 32.
Then, during the code search, after selecting each permutation, we constrain all following APMs by requiring the transition permutations $T_{ij}=F_jF_i^{-1}, F_j G_i^{-1}, G_j G_i^{-1}$ to commute with $A$, where only transitions between neighboring APMs are considered.
This guarantees that after we transform to the reference's cycle basis, the transitions reduce to  simple orbit-wise shifts and small inter-orbit permutations (see \cref{app:apm_compile}).
This introduces additional constraints, in addition to girth $\geq 6$ that we impose.
We perform post-search filtering using finite-size distance-bounding tools and simulations of logical error rates under code capacity noise.
Note that in general, the orbit-wise shifts are not guaranteed to be homogeneous between rows. In practice, this would lead to serialization overhead, as each row would have to be addressed one after another.
Nevertheless, we find that by modifying the initial offset $s_r$ of each orbit, we can minimize the number of independent row-shifts that need to be applied.
Furthermore, as explained in \cref{app:apm_compile}, some of the resulting codes have additional structures that guarantee homogeneous shifts between rows.
We search across a few hundred seeds for each set of parameters.
A more extensive search may yield further improvements to the code parameters.

Table~\ref{tab:apm-params} lists representative instances.
The notation $\llbracket n,k,d\leq d_{\mathrm{ub}}\rrbracket$ reports a best-found upper bound on the distance rather than an exact distance determination.
The first two rows with $P=96$ and $P=192$ correspond to the two finite-size examples discussed explicitly in Sec.~\ref{sec:construction}.
The following rows provide additional instances with larger distance.

\begin{table*}[t]
  \centering
  \caption{APM CSS code parameters for selected codes at $P = 96$, $192$, and $384$.
    All codes use $J = 3$ active block rows, $L = 12$ block columns, girth $\geq 6$,
    and non-commutativity pairs $(0,3)$, $(1,2)$.
  The affine permutation maps are $f_i(x) = a_i x + b_i \pmod{P}$ and $g_i(x) = c_i x + d_i \pmod{P}$ for $i = 0, \ldots, 5$.}
  \label{tab:apm-params}
  \resizebox{1.4\columnwidth}{!}{%
    \begin{tabular}{@{}ccccccc@{}}
      \toprule
      $P$ & $\llbracket n, k, d \rrbracket$ & Rate & $i$ & $f_i(x) = a_i x + b_i$ & $g_i(x) = c_i x + d_i$ \\
      \midrule
      \multirow{6}{*}{96} &  \multirow{6}{*}{\nkdninetysix{}} & \multirow{6}{*}{0.503}
      & 0 & $5x + 41$   & $61x + 15$ \\
      & & & 1 & $85x + 77$  & $x + 24$   \\
      & & & 2 & $73x + 66$  & $89x + 62$ \\
      & & & 3 & $x + 0$     & $25x + 22$ \\
      & & & 4 & $x + 72$    & $85x + 93$ \\
      & & & 5 & $37x + 9$   & $25x + 78$ \\
      \midrule
      \multirow{6}{*}{192}  & \multirow{6}{*}{\nkdoneninetwo{}} & \multirow{6}{*}{0.502}
      & 0 & $71x + 127$  & $163x + 165$ \\
      & & & 1 & $97x + 80$   & $55x + 183$  \\
      & & & 2 & $67x + 117$  & $167x + 79$  \\
      & & & 3 & $163x + 165$ & $139x + 41$  \\
      & & & 4 & $25x + 60$   & $109x + 78$  \\
      & & & 5 & $187x + 33$  & $31x + 27$   \\
      \midrule
      \multirow{6}{*}{192} & \multirow{6}{*}{$\llbracket 2304, 1156, \leq 16 \rrbracket$} & \multirow{6}{*}{0.502}
      & 0 & $77x + 131$  & $169x + 138$ \\
      & & & 1 & $145x + 68$  & $157x + 135$ \\
      & & & 2 & $133x + 33$  & $125x + 143$ \\
      & & & 3 & $37x + 57$   & $145x + 148$ \\
      & & & 4 & $37x + 57$   & $121x + 78$  \\
      & & & 5 & $49x + 60$   & $121x + 126$ \\
      \midrule
      \multirow{6}{*}{384} & \multirow{6}{*}{$\llbracket 4608, 2308, \leq 20 \rrbracket$} & \multirow{6}{*}{0.501}
      & 0 & $23x + 321$  & $91x + 231$  \\
      & & & 1 & $205x + 178$ & $265x + 204$ \\
      & & & 2 & $217x + 324$ & $197x + 294$ \\
      & & & 3 & $223x + 237$ & $289x + 176$ \\
      & & & 4 & $301x + 66$  & $85x + 318$  \\
      & & & 5 & $67x + 3$    & $139x + 111$ \\
     \midrule
      \multirow{6}{*}{384} & \multirow{6}{*}{$\llbracket 4608, 2308, \leq 22 \rrbracket$} & \multirow{6}{*}{0.501}
      & 0 & $337x + 212$     & $37x + 201$    \\
      & & & 1 & $257x + 288$  & $x + 96$    \\
      & & & 2 & $37x + 105$    & $181x + 13$     \\
      & & & 3 & $229x + 249$    & $5x + 33$ \\
      & & & 4 & $145x + 132$    & $277x + 261$    \\
      & & & 5 & $157x + 39$    & $361x + 282$    \\
      \bottomrule
    \end{tabular}
  }
\end{table*}

\section{Details of Simulations}
\label{app:simulation}

In this section, we describe the details of the simulation we conducted in the main text.

\subsection{Circuits and noise models}
\label{app:circ-noise-model}
We focus on the $X$/$Z$ basis quantum memory simulation with $r=32$ rounds. More precisely, we initialize all physical qubits into $\ket{+}$/$\ket{0}$ states, followed by $r$ rounds of both $X$ and $Z$ stabilizer measurements. After that, we measure all physical qubits in $X$/$Z$ basis to determine if any logical errors happen.

For the circuit-level simulation, we assume that stabilizers are measured using the coloration syndrome extraction method for arbitrary CSS codes~\cite{tremblay2022constant}, where $X$ and $Z$ stabilizers are measured independently. The coloration algorithm results in a CNOT depth-12 circuit for both $X$ and $Z$ stabilizers, making the total depth of the circuit 24. It is worth noting that the depth of 24 can be reduced by a factor of 2 to 4 through various methods, such as left-right block syndrome extraction~\cite{strikis2026high,xu2024constant,menon2025magic,bravyi2024high} or Bell-pair-based syndrome extraction~\cite{cross2024improved}. We assume that the physical qubits can be initialized directly to the states $\ket{0}$ and $\ket{+}$ and measured in the Pauli $Z$ and $X$ bases. After initialization and prior to measurement of each qubit, a single-qubit depolarizing error with strength $p_{\mathrm{circ}}$ is applied. Additionally, after each two-qubit gate, including CNOT, a two-qubit depolarizing error with strength $p_{\mathrm{circ}}$ is applied. We assume no idling errors, motivated by the long coherence times of neutral atom and trapped ion systems. We leave more detailed analysis of realistic noise channels, incorporating the effect of idling and atom loss, to future work.

For the phenomenological simulation, we assume that the stabilizers are measured directly via Pauli product measurement with a probability $p_{\mathrm{meas}}$ of obtaining a wrong result. After each round of stabilizer measurement, all data qubits are subject to single-qubit depolarizing errors with strength $p_{\mathrm{data}}$. In each round of syndrome extraction, each data qubit will undergo 6 CNOT operations, while each ancilla qubit will undergo 12 CNOT operations. However, only $2/3$ of all depolarizing errors on average will cause a measurement flip on ancilla qubits. Therefore, we assume in addition that
\begin{equation}
  \frac{p_{\mathrm{data}}}{p_{\mathrm{meas}}}=\frac{6}{12\times 2/3} = \frac{3}{4},
\end{equation} and use a single parameter $p_{\mathrm{phenom}} = p_{\mathrm{data}}$ to control the error strength for our simulation.

\subsection{Decoders}
\label{app:decoders}

As described in Sec.~\ref{sec:decoding}, our hierarchical decoder consists of three tiers. We describe the implementation details of each tier below.

\paragraph{$T_1$ (belief propagation).}
For the phenomenological simulations, we use a quaternary belief propagation (BP) decoder that operates on the GF(4) Tanner graph, decoding $X$ and $Z$ errors jointly~\cite{poulin2008on,kuo2022exploiting,kuo2025generalized}.
For the circuit-level simulations, we use a separate binary BP implementation that operates on a single-basis detector error model: although both $X$ and $Z$ stabilizers are measured, detectors are constructed in only one basis.
Both implementations support memory BP, serial and alternating scheduling, sum-product and min-sum message-passing rules, and adaptive damping.

\paragraph{$T_2$ (relay BP).}
For the phenomenological simulations, we implement a relay variant built on top of the quaternary BP decoder from $T_1$, following the approach of Ref.~\cite{muller2025improved}.
For the circuit-level simulations, we use the open-sourced relay decoder of Ref.~\cite{muller2025improved} with the parameters listed in Table~\ref{tab:relay-params}. Since $T_1$ has already executed BP, the relay decoder skips its initial BP phase. Note that almost all parameters are the default ones from the package. We choose not to tune the relay parameters aggressively to avoid the failures that converged but wrong as we observed from the phenomenological simulation.

\begin{table}[h]
  \centering
  \caption{Parameters used for the open-sourced relay BP decoder~\cite{muller2025improved} in circuit-level simulations.}
  \label{tab:relay-params}
  \begin{tabular}{@{}ll@{}}
    \toprule
    Parameter & Value \\
    \midrule
    \texttt{num\_sets}    & 300 \\
    \texttt{set\_max\_iter} & 60 \\
    \texttt{gamma0}       & 0.1 \\
    \texttt{alpha\_iteration\_scaling\_factor} & 1.0 \\
    \texttt{stop\_nconv}  & 1 \\
    \texttt{precision}    & F32 \\
    \texttt{pre\_iter}    & 0 \\
    \texttt{seed}         & 0 \\
    \bottomrule
  \end{tabular}
\end{table}

\paragraph{$T_3$ (integer-programming MLE).}
The final fallback tier solves an exact most-likely-error (MLE) decoding problem via integer programming using the Gurobi solver~\cite{landahl2011fault,cain2024correlated,gurobi}. To keep the runtime manageable, the integer program is formulated using syndromes from a single basis, matching the logical observables of interest.

\section{Hierarchical Decoder and Real-Time Quantum Error Correction}
\label{app:hierarchical_decode}

The hierarchical decoder of Sec.~\ref{sec:decoding} is motivated primarily by simulation throughput, but the same staged architecture can also be employed for real-time quantum error correction.
In a real-time setting, we require that the decoder stack must sustain the incoming syndrome rate without an unbounded backlog and that the average reaction time across the whole circuit does not bottleneck execution~\cite{delfosse2020hierarchical,battistel2023real,das2020scalable,skoric2023parallel}.

Consider a hierarchy with $k$ tiers, where tier $i$ takes an average decoding time $t_i$ per invocation and passes unresolved instances to tier $i+1$ with escalation probability $r_i$.
The fraction of syndrome rounds that reach tier $i$ is
\begin{align}
  q_1 = 1, \qquad q_i = \prod_{j=1}^{i-1} r_j \quad (i\geq 2).
\end{align}
If all tiers share a single processor, the decoder performs an average of $\bar t = \sum_i q_i t_i$ of work per syndrome round, and must therefore satisfy $\bar t \lesssim T_{\mathrm{QEC}}$, where $T_{\mathrm{QEC}}$ is the time per syndrome extraction cycle, to keep up with the incoming syndrome rate.
BP handles every round, relay BP resolves most remaining non-converging cases, and only a small tail requires integer programming.
If the runtime ratio of the late stages to the early stages is much smaller than the ratio of non-fallback cases to fallback cases, then the average runtime is dominated by the early stages.

To make the BP and relay-BP timing concrete, we extrapolate from the recent FPGA implementation of Ref.~\cite{maurer2025real}, which realizes BP and relay-BP decoders for the gross code on a single FPGA.
The design completes one BP iteration in two FPGA clock cycles, giving a reported iteration time of $24~\mathrm{ns}$ for the gross $\llbracket 144,12,12\rrbracket$ code with sliding-window width $W=d=12$.
Average iteration counts range from of order $10$ (relay BP) to of order $20$ (standard BP) at $p\sim 10^{-3}$, and the $X$-stabilizer branch alone occupies roughly $52\%$ of the chip's LUTs.
Because the design is essentially spatially unrolled over the windowed Tanner graph, the dominant scaling parameter is the number of nodes per window, which grows as $nW\delta$ at fixed check weight, where $\delta$ is the data qubit degree.
Taking the gross-code decoder as the baseline, define
\begin{align}
  F \equiv \frac{nW \delta_{\text{Kasai}}}{144\cdot 12\delta_{\text{BB}}}.
\end{align}
Since $\delta_{\text{Kasai}}=\delta_{\text{BB}}=6$, for the Kasai instances at $W=d$, $F_{\llbracket 1152,580\rrbracket}=8.0$ and $F_{\llbracket 2304,1156\rrbracket}\approx 18.7$.
The factor $F$ thus captures the increase in the number of edges in the Tanner graph at fixed check weight.
Secondary effects, such as increased iteration counts due to larger graph diameter or memory-bandwidth limitations from larger graph size, could introduce additional overhead beyond the linear-in-$F$ estimate; the FPGA runtimes in Table~\ref{tab:fpga} should therefore be understood as order-of-magnitude estimates.

As an estimate, assume the reference fabric is held fixed and the larger decoding graph is time-multiplexed over it, so the per-iteration FPGA latency scales linearly in $F$, yielding $t_{\mathrm{iter}}\approx 192~\mathrm{ns}$ and $448~\mathrm{ns}$ per sliding-window iteration for our two codes.
Combining this with the average BP iteration counts observed in our circuit-level simulations at $p=10^{-3}$ ($\bar N_{\mathrm{iter}}^{T_1}=6$ and $8$, respectively) and amortizing over the $d$ rounds in each $W=d$ window gives per-round $T_1$ FPGA runtimes of $\bar N_{\mathrm{iter}}^{T_1}\,t_{\mathrm{iter}}/d$; together with the measured per-round CPU runtimes of each tier, this populates Table~\ref{tab:fpga}.
Since we did not directly measure iteration counts for the relay-BP stage, the $T_2$ FPGA entries are obtained by scaling the $T_1$ FPGA estimate by the corresponding CPU $T_2/T_1$ runtime ratio, under the assumption that relay BP incurs a similar per-iteration cost on FPGA as BP.
Under this conservative scaling, the $T_1$ BP stage runs at $\sim\!100$ to $\sim\!260\,\mathrm{ns}$ per syndrome round, and the $T_2$ relay-BP stage at $\sim\!1$ to $\sim\!20\,\mu\mathrm{s}$ per round, representing about five orders of magnitude speedup over the CPU baseline.
A more optimistic reading that scales the fabric with $nW$ would give substantially smaller per-iteration latencies, although this ignores inter-FPGA routing and signaling overheads.

Finally, the $T_3$ integer-programming tier delivers accurate logical-error estimates in simulation but with per-round CPU runtimes of several seconds (Table~\ref{tab:fpga}).
Approximate-optimal decoders, including the Tesseract decoder~\cite{beni2025tesseract,grbic2026accelerating} and recent neural-network decoders that approach its accuracy~\cite{gu2026scalable}, retain near-optimal logical-error performance at substantially lower and more predictable latencies, and are natural candidates for bounded-latency fallbacks in place of exact integer programming.
Extrapolating the $\sim\!40~\mu\mathrm{s}$ per-cycle GPU latency of the neural decoder of Ref.~\cite{gu2026scalable} on the gross code by the same factor $F$ gives roughly $320~\mu\mathrm{s}$ and $720~\mu\mathrm{s}$ per syndrome round for our two codes~--- roughly four orders of magnitude faster than the exact integer-programming baseline.

Even without GPUs, the CPU-only Gurobi $T_3$ fallback can still be manageable with further improvements.
When a decoding problem triggers the integer-programming tier, syndrome extraction continues uninterrupted while the solver runs.
For the $\llbracket 1152,580\rrbracket$ code, the $2.9\,\mathrm{s}$ CPU solve time spans roughly $350$ SE rounds ($T_{\mathrm{QEC}}=8.3\,\mathrm{ms}$); with $q_3=0.003\%$, the probability that any of those rounds also escalates to $T_3$ is $1-(1-q_3)^{350}\approx 1\%$.
For the $\llbracket 2304,1156\rrbracket$ code, improvements to either the early-tier decoders accuracy or the fallback decoders runtime can ensure that no backlog forms, as there will be an exponentially decreasing probability of producing another $T_3$ call.

Assuming FPGA acceleration for $T_1$ and $T_2$ and GPU acceleration for $T_3$ (the accelerator columns of Table~\ref{tab:fpga}), and weighting each tier's per-round time by its escalation fraction $q_i$, the total decoding work per round is $\bar t = \sum_i q_i t_i \approx 120\,\mathrm{ns}$ for the \nkdninetysix{} code and $\approx 1.0\,\mu\mathrm{s}$ for the \nkdoneninetwo{} code, well below the respective $T_{\mathrm{QEC}} = 8.3\,\mathrm{ms}$ and $9.9\,\mathrm{ms}$ syndrome-extraction cycles on our target neutral-atom architecture.

\begin{table*}[t]
  \centering
  \caption{Measured CPU runtimes and extrapolated accelerator runtimes for the three tiers of the hierarchical decoder on the two Kasai instances at $p=10^{-3}$. All runtimes are per QEC syndrome round (amortized over a $W=d$ sliding window), and $F\equiv nW/(144\cdot 12)$. $q_i$ is the fraction of syndrome rounds that reach tier $i$, from our circuit-level simulations; $q_1=1$ by construction. $T_1$ and $T_2$ FPGA entries extrapolate from the FPGA decoder of Ref.~\cite{maurer2025real}, using $\bar N_{\mathrm{iter}}^{T_1}\,t_{\mathrm{iter}}/d$ for $T_1$ and the observed CPU $T_2/T_1$ ratio for $T_2$. The $T_3$ GPU entry extrapolates the neural-network decoder of Ref.~\cite{gu2026scalable}, scaling its reported $\sim\!40\,\mu\mathrm{s}$ per QEC cycle on the gross $\llbracket 144,12,12\rrbracket$ code by $F$. $T_{\mathrm{QEC}}$ is the syndrome-extraction cycle time for each code on the neutral-atom architecture considered in the main text.}
  \label{tab:fpga}
  \resizebox{\textwidth}{!}{%
    \begin{tabular}{@{}lccccccccccc@{}}
      \toprule
      & & \multicolumn{3}{c}{$T_1$ (BP)} & \multicolumn{3}{c}{$T_2$ (relay BP)} & \multicolumn{3}{c}{$T_3$} & \\
      \cmidrule(lr){3-5} \cmidrule(lr){6-8} \cmidrule(lr){9-11}
      Code & $F$ & $\bar N_{\mathrm{iter}}^{T_1}$ & CPU & FPGA & $q_2$ & CPU & FPGA & $q_3$ & CPU & GPU & $T_{\mathrm{QEC}}$ \\
      \midrule
      \nkdninetysix{} & $8.0$  & $6$ & $18\,\mathrm{ms}$ & $100\,\mathrm{ns}$ & $0.8\%$                   & $0.20\,\mathrm{s}$ & $1.1\,\mu\mathrm{s}$ & $0.003\%$                  & $2.9\,\mathrm{s}$ & $320\,\mu\mathrm{s}$ & $8.3\,\mathrm{ms}$  \\
      \nkdoneninetwo{}  & $18.7$ & $8$ & $47\,\mathrm{ms}$ & $260\,\mathrm{ns}$ & $1.3\%$                   & $3.1\,\mathrm{s}$  & $17\,\mu\mathrm{s}$  & $0.07\%$                   & $10\,\mathrm{s}$  & $720\,\mu\mathrm{s}$ & $9.9\,\mathrm{ms}$ \\
      \bottomrule
    \end{tabular}%
  }
\end{table*}

\section{Heuristic finite-blocklength estimates}
\label{app:finite_blocklength}

In this appendix, we derive the finite-blocklength guide curves used in Sec.~\ref{sec:pareto} and Fig.~\ref{fig:pareto}. The construction is the direct quantum analog of the classical normal approximation to the finite blocklength channel coding rate~\cite{polyanskiy2010channel,tomamichel2016quantum}, applied heuristically with the quantum hashing bound used to approximate capacity.

The main result of Ref.~\cite{polyanskiy2010channel} states that the finite blocklength coding rate can be approximated by
\begin{equation}
  \label{eq:fb_classical}
  R^{*}(n,p,\varepsilon) = C(p) - \sqrt{\frac{V(p)}{n}}\,Q^{-1}(\varepsilon)
  + O\!\left(\frac{\log n}{n}\right),
\end{equation}
where $n$ is the block length, $p$ the channel error rate, $\varepsilon$ the logical error probability, $C$ is the channel capacity, $V$ is the channel dispersion measuring the stochastic variability of the channel relative to a deterministic channel, and $Q(x)=\tfrac{1}{\sqrt{2\pi}}\int_x^{\infty} e^{-t^2/2}\,dt = 1-\Phi(x)$ is the Gaussian tail. The $\sqrt{V/n}$ term is a finite-size ``backoff'' from capacity: it arises because the information density accumulated over $n$ uses is, by the central limit theorem, approximately Gaussian about its mean $C$ with variance $V/n$. Eq.~\eqref{eq:fb_classical} is known to be remarkably accurate already at blocklengths of a few hundred.

Rigorous generalizations of these bounds to quantum codes have been derived for the dephasing channel and erasure channel~\cite{tomamichel2016quantum}. Instead, we choose to focus on the case of the depolarizing channel, which serves as a better proxy for realistic noise. Given the lack of tight estimates in this setting, we instead apply the formulae heuristically, based on the probability distribution of the relevant error variable.
For the depolarizing channel, the per-qubit Pauli error distribution is
\begin{equation}
  \label{eq:pauli_dist}
  P_{\mathrm{dep}}(p) = (P_I,P_X,P_Y,P_Z)
  = \Bigl(1-p,\ \tfrac{p}{3},\ \tfrac{p}{3},\ \tfrac{p}{3}\Bigr).
\end{equation}
The quantum Hashing bound provides a lower bound (achievability bound) on the quantum capacity that can be realized with a random stabilizer code, with the rate given as
\begin{align}
  \label{eq:hashing}
  R_H(p) &= 1 - H\!\bigl(P_{\mathrm{dep}}(p)\bigr)\nonumber\\
  &=1+\underbrace{p\log_2 p + (1-p)\log_2(1-p)}_{-h(p)} - p\log_2 3.
\end{align}

To obtain the finite-size correction to the Hashing rate, we calculate the variance of the capacity for a given realization of the error.
For an i.i.d.\ error string $\mathbf{e}=(e_1,\dots,e_n)$, the normalized surprisal (characterizing the information gained upon the observation) is given by
\begin{equation}
  s_n(\mathbf{e}) \equiv -\frac{1}{n}\log_2 P(\mathbf{e})
  = -\frac{1}{n}\sum_{i=1}^{n}\log_2 P_{\mathrm{dep}}(e_i).
\end{equation}
This is an average of $n$ i.i.d.\ random variables with mean
$\mathbb{E}\!\left[-\log_2 P_{\mathrm{dep}}(e_i)\right]
= 1-R_H(p)$ and per-symbol variance
\begin{align}
  \label{eq:varentropy}
  V(p) &= \mathrm{Var}_{a\sim P_{\mathrm{dep}}}\!\bigl[-\log_2 P_{\mathrm{dep}}(a)\bigr]\nonumber\\
  &= \sum_{a} P_a\,(\log_2 P_a)^2 - H\!\bigl(P_{\mathrm{dep}}\bigr)^2,
\end{align}
the \emph{varentropy} of the Pauli distribution, which here plays the role of the channel dispersion. By the central limit theorem
$s_n \approx \mathcal{N}\!\bigl(1-R_H(p),\,V(p)/n\bigr)$. The hashing argument
identifies an error as correctable when it is sufficiently typical,
$s_n \le 1-R$ with $R=k/n$ the code rate, while atypically improbable errors
($s_n>1-R$) dominate the residual failure probability. Equating this failure
probability to the target $\varepsilon$,
\begin{equation}
  \varepsilon \approx \Pr\!\bigl[\,s_n > 1-R\,\bigr]
  = Q\!\left(\frac{R_H(p) - R}{\sqrt{V(p)/n}}\right),
\end{equation}
and solving for $R$ gives the maximal heuristic rate
\begin{equation}
  \label{eq:fb_quantum}
  \,R^{*}(n,p,\varepsilon) \approx R_H(p)
  - \sqrt{\frac{V(p)}{n}}\,\Phi^{-1}(1-\varepsilon)\,,
\end{equation}
which is Eq.~\eqref{eq:finite} of the main text. Evaluating
\eqref{eq:varentropy} with the distribution \eqref{eq:pauli_dist},
we have
\begin{equation}
  \label{eq:Vexplicit}
  V(p) = (1-p)\bigl(\log_2(1-p)\bigr)^2
  + p\left(\log_2\tfrac{p}{3}\right)^2
  - \bigl(h(p) + p\log_2 3\bigr)^2.
\end{equation}

Our simulations in Fig.~\ref{fig:pareto} use the preceding heuristic estimates for the depolarizing channel. In Fig.~\ref{fig:pareto}(a), we adjust the target $\varepsilon$ used in the estimate to be $k\times$ larger, as the logical error rate target for the data points is per logical operator.
In Fig.~\ref{fig:pareto}(b), we use the distance 12 instance for $P=96$, the distance 16 instance for $P=192$ and the distance 20 instance for $P=384$ described in Tab.~\ref{tab:apm-params}.

\section{Compiling Affine Permutations into AOD Moves}
\label{app:apm_compile}

The syndrome-extraction schedule applies affine permutations $x \mapsto ax + b$ on $\mathbb{Z}_P$ to the ancilla block between successive CNOT layers.
AOD transport implements only \emph{separable} (product) rearrangements, in which the new row position depends only on the old row and the new column position only on the old column, so a generic $P\times P$ permutation cannot be realized in a single step.
In this appendix, we describe conditions under which $\mathbb{Z}_P$ can be laid out on a 2D grid so that every affine permutation becomes separable, and further describe structures that reduce the reconfiguration cost along each axis.

\Cref{app:crt} shows that whenever $P = ml$ with $\gcd(m,l)=1$, the Chinese Remainder Theorem (CRT) provides a layout in which every APM is separable into independent row and column movements. \Cref{app:column_structure} then examines the column components of the APMs considered in this work, explaining how the commutation structure simplifies the movement, and identifies structures that reduce the movements for the \nkdninetysix{} and \nkdoneninetwo{} instances down to global cyclic shifts. \Cref{app:log_move} closes with a log-depth fallback for general APMs lacking either structure.

\subsection{CRT Representation and Separability}
\label{app:crt}

Suppose $P = ml$ with $\gcd(m,l) = 1$. The Chinese Remainder Theorem gives a \emph{ring} isomorphism
\[
  \mathbb{Z}_P \;\cong\; \mathbb{Z}_m \times \mathbb{Z}_l,
  \qquad x \;\longmapsto\; (x \bmod m,\, x \bmod l),
\]
under which an affine permutation $f(x) = ax + b$ acts independently on the two factors as $f_m(x_m) = (a \bmod m)\, x_m + (b \bmod m)$ on $\mathbb{Z}_m$ and $f_l(x_l) = (a \bmod l)\, x_l + (b \bmod l)$ on $\mathbb{Z}_l$. CRT thus reduces a single APM on $\mathbb{Z}_{ml}$ to two independent APMs, one on $\mathbb{Z}_m$ and one on $\mathbb{Z}_l$.

For AOD hardware, this decomposition is exactly the separability constraint: a crossed-AOD trap array has one AOD controlling row positions ($\mathbb{Z}_m$) and another controlling column positions ($\mathbb{Z}_l$), so each $P$-qubit block occupies an $m \times l$ atom sub-array with qubit $x$ at cell $(x \bmod m,\, x \bmod l)$, and the two APM factors become global column or row operations. If, additionally, the column-factor APMs all pairwise commute, they generate an abelian group acting on $\mathbb{Z}_l$, and after a one-time relabeling of columns each step becomes a rigid cyclic shift (\cref{app:column_structure}).

Intuitively, the layout can be thought of as consisting of a diagonal that wraps around the torus, since both coordinates increment together as $x$ runs through $\mathbb{Z}_P$. The layout is \emph{not} a row-major reshape of $\mathbb{Z}_P$, which can mix the two coordinates (\cref{app:log_move}).

All code instances in this work have $P = 3\cdot 2^k$ with $k \in \{5,6,7,8\}$, suggesting a natural layout on a $3 \times 2^k$ grid. For example, $g_0(x) = 61 x + 15$ on $\mathbb{Z}_{96}$ (the \nkdninetysix{} instance of \cref{tab:apm-params}) reduces to the identity on $\mathbb{Z}_3$ (row factor) and to $x_{32} \mapsto 29\, x_{32} + 15$ on $\mathbb{Z}_{32}$ (column factor), such that implementing $g_0$ requires no row motion at all.
For the $P=192$ instance, we instead choose a layout of dimensions $6\times 32$ instead of $3\times 64$, based on the group structure of the column factor.

\subsection{Designing APMs to Simplify Atom Movement}
\label{app:column_structure}

The CRT representation decouples row and column movements, but leaves open their individual costs. Row work on the $\mathbb{Z}_3$ factor permutes at most 3 rows and is cheap. The compilation question therefore reduces to the structure of the 12 \emph{column components} $\bar m_i \in \mathrm{Aff}(\mathbb{Z}_{P/3})$, where $\bar m_i(x) = (a_i \bmod P/3)\, x + (b_i \bmod P/3)$ is the $\mathbb{Z}_{P/3}$ reduction of the $i$th Kasai generator in \cref{eq:kasai_x,eq:kasai_z}. The physical AOD moves between rounds are transitions $\bar F_j \bar F_i^{-1}$ rather than the raw $\bar m_i$; these transitions lie in the group generated by the column components, so analyzing that group controls both.

We first make precise the design principle described in Sec.~\ref{sec:construction}: we constrain transitions to commute with a chosen reference APM \(A\) with large orbits.

\begin{lemma}
  Let \(A\) be an APM acting on $\mathbb{Z}_P$, and write its orbits as
  \[
    \mathcal O_i=\{A^t\gamma_i : t\in \mathbb Z\},
  \]
  where \(\gamma_i\) is a representative of the \(i\)-th orbit. Let \(M\) be another APM such that
  \[
    MA=AM.
  \]
  Then for each orbit \(\mathcal O_i\), there exist an orbit \(\mathcal O_j\) and an integer \(s_i\) such that
  \[
    M(A^t\gamma_i)=A^{t+s_i}\gamma_j
    \qquad\text{for all } t\in\mathbb Z.
  \]
  Equivalently, \(M\) maps each orbit of \(A\) onto another orbit of \(A\), and acts within that orbit as a uniform shift.
\end{lemma}

\begin{proof}
  Fix an orbit representative \(\gamma_i\). $M\gamma_i$ lies in some orbit of \(A\), say
  \[
    M\gamma_i=A^{s_i}\gamma_j
  \]
  for some \(j\) and some integer \(s_i\). Using the commutation relation \(MA=AM\), we then have for every \(t\in\mathbb Z\),
  \[
    M(A^t\gamma_i)=A^tM\gamma_i=A^tA^{s_i}\gamma_j=A^{t+s_i}\gamma_j.
  \]
  Thus every point in the orbit \(\mathcal O_i\) is mapped into the orbit \(\mathcal O_j\), with the same offset \(s_i\). Since \(M\) is a permutation, this map is bijective from \(\mathcal O_i\) onto \(\mathcal O_j\). This proves the claim.
\end{proof}

This guarantees that each transition decomposes, in $A$-orbit coordinates, as a cyclic shift within each orbit plus a permutation between orbits. However, for AOD compilation we would like more: for the shifts across orbits to coincide, so that the entire move is one global cyclic shift along the column axis, up to a small row permutation. The following structural statement identifies when this happens for the codes in Tab.~\ref{tab:apm-params} and Kasai's original $P=768$ code~\cite{kasai2026breaking}.

\begin{proposition}[Maximal abelian subgroup]
  \label{thm:max_abelian}
  Let $G_{\mathrm{col}} = \langle \bar m_1, \ldots, \bar m_{12}\rangle \leq \mathrm{Aff}(\mathbb{Z}_{P/3})$ be the subgroup generated by the 12 column components in our code instances. For each of the code instances we consider, the maximal abelian subgroup $B$ of $G_{\mathrm{col}}$ has order exactly $P/3$, verified by direct computation. $G_{\mathrm{col}}$ itself is abelian for all new codes identified in this paper, and non-abelian for Kasai's original construction.
\end{proposition}

$B$ provides a convenient basis transformation that simplifies atom movement and ensures that all rows experience the same shift when a column component lies in this subgroup.
One can directly verify that $B$ acts regularly on $\mathbb{Z}_{P/3}$, so $b \mapsto b(0)$ is a bijection $B \leftrightarrow \mathbb{Z}_{P/3}$.
By the structure theorem for finite abelian groups, $B$ has cyclic generators $g_1, \ldots, g_r$ with $B \cong \mathbb{Z}_{n_1} \oplus \cdots \oplus \mathbb{Z}_{n_r}$, so every element of $B$ factors uniquely as $g_1^{e_1} \cdots g_r^{e_r}$. The bijection lifts this factorization to $\mathbb{Z}_{P/3}$: each column $c$ acquires a unique exponent vector $(e_1, \ldots, e_r)$ with $g_1^{e_1} \cdots g_r^{e_r}(0) = c$. Loading atoms into the AOD in this exponent order ensures that every $\bar m_i$ acts as a quasi-cyclic translation on $B$.
Note that this ordering is a software convention at startup; no atoms move as part of the relabeling.

More specifically, for the $P = 96$ instance, $B$ is cyclic and each round compiles to a single global cyclic shift; for the first $P = 192$ and $P=384$ instances, $B \cong \mathbb{Z}_{32} \times \mathbb{Z}_2$ and each round compiles to a pair of commuting shifts on the two factors, while for the latter $P = 192$ and $P=384$ instances, $B$ is cyclic. For Kasai's $P = 768$ instance, only some of the 12 column components lie inside a $P/3$-sized abelian subgroup and can be implemented with a global shift; the remaining ones may require a limited number of row-dependent cyclic shifts.
We now explain these cases in detail.

We verify commutation of two affine maps $m_i(x) = a_i x + b_i$ and $m_j(x) = a_j x + b_j$ in $\mathrm{Aff}(\mathbb{Z}_M)$ with the offset condition
\begin{equation}
  \label{eq:commute-condition} (a_i - 1)\, b_j \;\equiv\; (a_j - 1)\, b_i \pmod M.
\end{equation}

\paragraph{$P = 96$: cyclic.}
The subgroup generated by the 12 column components is the cyclic group $\mathbb{Z}_{32}$, where every column component can be written as some power $\gamma^{s_i}$ of a generator $\gamma$. Indeed, the non-commuting components of this code are fully concentrated in the $\mathbb{Z}_3$ portion.

Define the column relabeling $\sigma : \mathbb{Z}_{32} \to \mathbb{Z}_{32}$ by $\sigma(k) = \gamma^k(0)$; this is a bijection because $G_{\mathrm{col}}$ acts regularly on $\mathbb{Z}_{32}$. Under the relabeling $c = \sigma(k)$, every Kasai column component acts as a rigid cyclic shift:

\begin{lemma}[Linearization, $P=96$]
  \label{lem:linearize_p96}
  For each $i$, $\bar m_i(\sigma(k)) = \sigma(k + s_i)$, so in $k$-coordinates $\bar m_i$ sends $k$ to $k + s_i \pmod{32}$.
\end{lemma}
\begin{proof}
  $\bar m_i = \gamma^{s_i}$, so $\bar m_i(\sigma(k)) = \gamma^{s_i}(\gamma^k(0)) = \gamma^{s_i + k}(0) = \sigma(k + s_i)$, where the middle equality is the definition of group powers.
\end{proof}

The relabeling is a one-time choice of atom loading order at startup: physical column slot $k$ holds the qubit whose original $\mathbb{Z}_{96}$ column coordinate is $\sigma(k)$, and no atoms move as a result. Each subsequent round's column-AOD work is a single rigid cyclic shift by the integer $s_i$ for the current map. The code is \emph{quasi-cyclic} on its column factor: $\Phi : (r, k) \mapsto (r, k+1)$ commutes with every column and row action of the 12 generators, so it is a code automorphism of order 32, partitioning the 1152 data qubits into 36 orbits of length 32 and making the parity-check matrix in the relabeled basis a block matrix of $32\times 32$ circulants.

\paragraph{$P = 192$ instance 1: bivariate cyclic.}
Closing the 12 generators under composition yields $|G_{\mathrm{col}}| = 64$, but in contrast to $P = 96$ the group is \emph{not} cyclic: the maximum element order inside $G_{\mathrm{col}}$ is 32. Instead,
\[
  G_{\mathrm{col}} \;\cong\; \mathbb{Z}_{32} \times \mathbb{Z}_2,
\]
with explicit generators
\begin{align*}
  \gamma_2(c) &= 13\, c + 30 \pmod{64}, \qquad \text{order } 32,\\
  \gamma_1(c) &= -c + 43 \pmod{64}, \qquad \text{order } 2,
\end{align*}
Here $\gamma_2$ has order 32 with two orbits of length 32 on $\mathbb{Z}_{64}$ and plays the role of the reference APM at this block size. The element $\gamma_1$ is an affine reflection of $\mathbb{Z}_{64}$ about the half-integer point $c = 43/2$; it commutes with $\gamma_2$ and swaps the two $\gamma_2$-orbits, capturing in the bivariate structure what we would label as an inter-orbit permutation from the point of view of a single reference APM. Every Kasai column component factors uniquely as $\bar m_i = \gamma_1^{a_i} \gamma_2^{b_i}$ with $(a_i, b_i) \in \mathbb{Z}_2 \times \mathbb{Z}_{32}$.

Despite there being multiple orbits, the product group structure helps ensure that the same shift is applied to both, ensuring a global column cyclic shift.
Defining $\sigma(a, b) = \gamma_1^a \gamma_2^b(0) : \mathbb{Z}_2 \times \mathbb{Z}_{32} \to \mathbb{Z}_{64}$ (again a bijection by regularity), the same one-line calculation as in \cref{lem:linearize_p96} gives $\bar m_i(\sigma(a, b)) = \sigma(a + a_i, b + b_i)$, so every $\bar m_i$ acts as a rigid translation on $\mathbb{Z}_2 \times \mathbb{Z}_{32}$. The parity-check matrix in the relabeled basis decomposes into blocks that are circulant in two variables, with the second variable coming from a reflection rather than a second independent translation. Each round's column-AOD work is a pair of commuting rigid shifts, one on $\mathbb{Z}_{32}$ and one on $\mathbb{Z}_2$, with no per-column frequency reprogramming.

\paragraph{Remaining $P=192$ instance and $P=384$ instances}
These code instances are all constructed to have a maximal Abelian subgroup with order $P/3$. For the second $P=192$ instance, the group is $\mathbb{Z}_{64}$; for the first $P=384$ instance, the group is $\mathbb{Z}_2\times\mathbb{Z}_{64}$ while for the second $P=384$ instance, the group is $\mathbb{Z}_{128}$.

\paragraph{$P = 768$.}
This is Kasai's original $\llbracket 9216, 4612, \leq 48\rrbracket$ instance~\cite{kasai2026breaking}. The group $G_{\mathrm{col}}$ is non-abelian, and its maximal abelian subgroup has order $P/3 = 256$. Seven generators lie inside, $\{\bar f_4, \bar g_0, \bar g_1, \bar g_2, \bar g_3, \bar g_4, \bar g_5\}$, again the ones with $v_2(b_i) \geq 2$.

\paragraph{Summary and consequences.}
In each of these code instances we consider, the maximal abelian subgroup has order $P/3$; the fraction of Kasai generators lying inside it and the internal structure of the subgroup itself vary. \Cref{tab:family_summary} summarizes these patterns.

\begin{table*}[t]
  \centering
  \caption{Structure of the maximal abelian subgroup of the Kasai column group $G_{\mathrm{col}}$ across the four $P = 3\cdot 2^k$ variants considered in this work. ``\# inside'' counts the number of the 12 Kasai column components that lie in the maximal abelian subgroup. The maximal abelian subgroup itself always has order exactly $P/3$.}
  \label{tab:family_summary}
  \renewcommand{\arraystretch}{1.2}
  \begin{tabular}{c|c|c|l}
    \hline
    $P$ & \# inside & structure & outside generators \\
    \hline
    96  & $12/12$ & $\mathbb{Z}_{32}$ (cyclic) & none \\
    192 & $12/12$ & $\mathbb{Z}_{32} \times \mathbb{Z}_2$ or $\mathbb{Z}_{64}$& none \\
    384 & $12/12$  & $\mathbb{Z}_{64} \times \mathbb{Z}_2$ or $\mathbb{Z}_{128}$ & none \\
    768 & $7/12$  & abelian, rank $\geq 2$ & $\bar f_0, \bar f_1, \bar f_2, \bar f_3, \bar f_5$ \\
    \hline
  \end{tabular}
\end{table*}

The \nkdninetysix{} and \nkdoneninetwo{} instances have all 12 Kasai generators belonging to the maximal abelian subgroup. At $P = 96$, this reduces each round's column-AOD work to one global rigid cyclic shift; at $P = 192$, to a pair of commuting shifts. When we lay out the $P=192$ code as $6\times 32$, this naturally leads to global cyclic shifts in the column direction and swaps in the row direction. We note that these conclusions are conservative: raw-APM abelianness is sufficient for rigid-shift compilation but not strictly necessary, since the schedule-specific transitions $\bar F_j \bar F_i^{-1}$ could commute pairwise even when the raw-APM group does not. Whether this weaker condition holds depends on the specific syndrome-extraction schedule.

\subsection{Log-Depth Fallback for Generic Affine Permutations}
\label{app:log_move}

For affine permutations with no coprime factorization (\cref{app:crt}) or exploitable commuting structure (\cref{app:column_structure}), we describe a generic fallback that implements an arbitrary affine permutation in logarithmic depth on a 2D AOD grid.
In contrast to the 1D layout of each permutation matrix sub-block used in Ref.~\cite{xu2024constant}, our scheme is compatible with a 2D qubit layout, removing the need for an extra conversion layer between horizontal and vertical 1D layout and improving the aspect ratio of the computation region.

We implement each data or ancilla block with a 2D matrix of size $n=m\times l$, with the $i$th qubit having coordinates $(x,y)$, where $i=xm+y$.
Denoting the coordinates before and after the affine permutation with subscripts $1$ and $2$, the affine permutation maps $i_1\rightarrow i_2=Ai_1+B\left(\textrm{mod}\;N\right)$, which we can rewrite coordinate-wise as
\begin{align}
  x_2&=Ax_1+B\left(\textrm{mod}\;m\right),\label{eq:x}\\
  y_2&=\left[(Ax_1+B)//m+Ay_1\right]\left(\textrm{mod}\;l\right),\label{eq:y}
\end{align}
where $//$ denotes the integer quotient from division.

We implement the required atom movement in 3 stages (Fig.~\ref{fig:generic_apm}), requiring in total at most $2\log m+4\log l$ AOD grid moves.

First, we perform the rearrangement in the $x$ direction.
As seen in \cref{eq:x}, this rearrangement is independent of the $y$ coordinate, and can be performed in parallel within the same APM block.
Utilizing the divide and conquer 1D arbitrary permutation algorithm in Ref.~\cite{xu2024constant}, this can be implemented with a single AOD grid in at most $2\log m$ layers.

Next, we perform the rearrangement in the $y$ direction.
The affine component $Ay_1$ can be implemented in the same way as the $x$ direction using $2\log l$ layers.

Finally, the qubit ordering also introduces a column-dependent cyclic shift of the $y$ coordinate, as seen in the first term of \cref{eq:y}.
Crucially, the amount of the shift is only dependent on $x_1$ and not $y_1$, indicating that it is cyclic and has fewer degrees of freedom than a per-column arbitrary permutation.
We implement this in $\log l$ batched cyclic shifts ($2\log l$ individual grid moves), by decomposing each cyclic shift into a binary representation, and performing shifts on all columns that have a given power of two in their decomposition together.
For example, cyclic shifts of $[1,7,9,5,4,2]$ across 6 columns decompose into shifts by 1, 2, 4, and 8 applied to different column subsets, four steps in total.

While we have focused on the case of general APMs, the same algorithm also works for cyclic permutations, which underlie the core structure of lifted product (LP) codes.
The ring elements used to construct LP codes are special instances of APMs with $A=1$.
In this case, the shifts from $A$ become trivial, and instead of performing a more complex permutation, we simply perform a cyclic shift.
In total, we need to perform a global cyclic shift in the $x$ direction (2 AOD moves), a global cyclic shift in the $y$ direction (2 AOD moves), and a cyclic shift of a subset of the columns in the $y$ direction due to the carry bit (2 AOD moves).
Therefore, a single cyclic permutation in 2D requires 6 AOD moves.
When $l$ and $m$ are coprime, the CRT representation of \cref{app:crt} applies and removes the column-dependent moves.

\begin{figure*}[t]
  \centering
  \includegraphics[width=2\columnwidth]{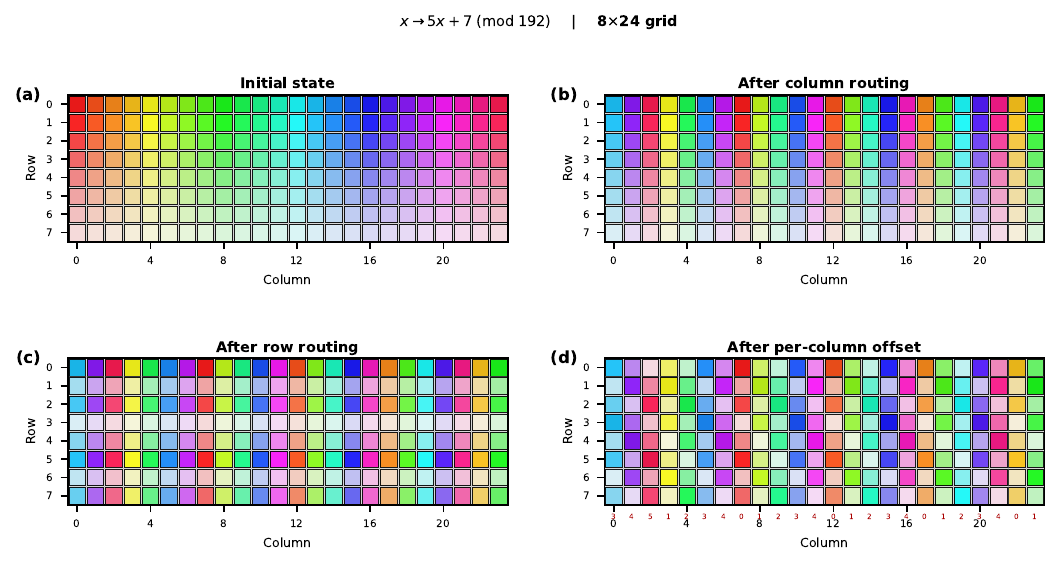}
  \caption{Implementing a generic APM in log depth. (a) Initial state with qubits color-coded. (b) Global routing of different columns using a log-depth 1D permutation in the $x$ direction. (c) Global routing of different rows using a log-depth 1D permutation in the $y$ direction. (d) Per-column offset, implemented in log-depth by decomposing the shifts into binary representation.}
  \label{fig:generic_apm}
\end{figure*}

\subsection{Atom Layout and Time Estimates}
\label{app:precise_layout}

\begin{table*}[t]
  \centering
  \begin{tabular}{l c c}
    \hline
    & \nkdninetysix{} & \nkdoneninetwo{} \\
    \hline
    Space (width $\times$ height) & $374\,\mu\mathrm{m} \times 420\,\mu\mathrm{m}$ & $374\,\mu\mathrm{m} \times 852\,\mu\mathrm{m}$ \\
    SE round (2 crossed-AODs) & $13{,}267\,\mu\mathrm{s}$ & $16{,}910\,\mu\mathrm{s}$ \\
    SE round (4 crossed-AODs) & $8{,}317\,\mu\mathrm{s}$ & $9{,}935\,\mu\mathrm{s}$ \\
    \hline
  \end{tabular}
  \caption{Space requirements (excluding measurement zones) and estimated SE round time for the \nkdninetysix{} and \nkdoneninetwo{} codes.}
  \label{tab:space-time-summary}
\end{table*}

We model the execution of the \nkdninetysix{} and \nkdoneninetwo{} codes on a neutral-atom quantum computer and estimate the time per syndrome extraction (SE) round. We make the following assumptions:
\begin{enumerate}
  \item Data atoms are spaced $12\,\mu\mathrm{m}$ apart both vertically and horizontally; ancilla atoms are placed $2\,\mu\mathrm{m}$ from the corresponding data atom for each gate~\cite{bluvstein2022quantum,zhou2025resource}.
  \item Atoms are moved smoothly with constant acceleration and deceleration of $5500\,\mathrm{m/s^2}$~\cite{bluvstein2022quantum}.
  \item Atom moves are modeled as straight lines. The duration of straight-line moves is a good approximation to the slightly curved paths needed to avoid atom collisions.
  \item A parallel round of CZ gates takes $1\,\mu\mathrm{s}$, which is negligible compared to move times of hundreds of microseconds.
  \item A parallel round of Hadamard gates is likewise negligible.
  \item Measurement of the $Z$ and $X$ checks is staggered as described in the main text and takes as little as $500\,\mu\mathrm{s}$~\cite{bluvstein2022quantum}.
  \item Rather than waiting for measurement to complete before resuming gates, we bring in fresh qubits to replace the ancilla qubits and continue gates in parallel with measurement. This makes the measurement time irrelevant in our regime, as move speed is the bottleneck.
  \item We operate either two or four crossed-AODs in parallel to independently control two or four sets of dynamic traps. Two AODs allow shifts or swaps to be trivially parallelized, reducing the total time by nearly half. Four AODs additionally enable data-block moves to occur in parallel with the ancilla-qubit permutations, reducing the total time by roughly another factor of two. The net effect is that any two permutations shown in a single frame of Figs.~\ref{fig:P96-full-move-sequence} and~\ref{fig:P192-full-move-sequence} are performed simultaneously.
  \item With four AODs, up to two of the permutations indicated by arrows in the figures can be executed concurrently. We choose the ordering freely to minimize total time. Note that a pair of row swaps (one $Z$ and one $X$) counts as a single operation for the purpose of parallelization, because two AODs can handle any number of non-conflicting swaps simultaneously.
\end{enumerate}

\paragraph{Results.}
We find the APM ordering $F_0, F_5, F_2, F_4, F_1, F_3, G_3, G_1, G_4, G_2, G_5, G_0$ to yield favorable total time for both codes.
Any reordering of the APMs is valid provided all $F$ APMs precede all $G$ APMs.
We leave the co-optimization of movement time and hook error mitigation to future work.
We further require the first $F$ and the last $G$ to share the same index, and likewise for the last $F$ and first $G$.
Table~\ref{tab:space-time-summary} summarizes the space and time cost of this ordering; the detailed atom positions and rearrangement sequences are shown in Figs.~\ref{fig:P96-full-move-sequence} and~\ref{fig:P192-full-move-sequence}.

We also note that this code construction is particularly well-suited for using Bell pairs to parallelize syndrome extraction, since such a choice will match the number of data and ancilla qubits, allowing maximal parallelism where each data and ancilla qubit is involved in an entangling gate in each layer. To perform $X$ and $Z$ syndrome extraction in parallel, we can use the left-right circuit approach~\cite{xu2024constant,strikis2026high} or the macroscheduler approach~\cite{menon2025magic}. Bell checks can be readily incorporated into these schedules, with no or only a small increase in circuit depth.

\subsection{Movement Compilation for Lifted Product Codes}
\label{app:lifted_product}

The compilation strategy of \cref{app:apm_compile} is not specific to the Kasai construction; it applies directly to lifted product (LP) codes~\cite{panteleev2019degenerate} built over a cyclic group ring, a family that includes the bivariate bicycle (BB) codes~\cite{bravyi2024high,yoder2025tour}. Such a code is specified by base matrices over the group algebra $R=\mathbb{F}_2[\mathbb{Z}_\ell]$ of a cyclic group, in which each scalar entry is a polynomial whose monomials $x^s$ are the cyclic shifts of $\mathbb{Z}_\ell$. The resulting CSS code carries a free $\mathbb{Z}_\ell$ action, so each data and ancilla qubit is labeled by a (small) base index together with a lift index in $\mathbb{Z}_\ell$. Every monomial appearing in a check therefore acts on the lift index as a single uniform cyclic shift $t\mapsto t+s$, independent of the base index: in the language of \cref{app:apm_compile}, the ring elements of an LP code are exactly the affine permutations with $a=1$, corresponding to pure translations.

A cyclic shift of a single lift block is by itself inexpensive, but realizing it as one crossed-AOD move requires that the \emph{same} shift act simultaneously on every block that shares the moving axis; if different blocks required different shifts, the AOD would have to address them one at a time, reintroducing the per-row serialization noted in \cref{app:code_construction}. The product structure of the (lifted) hypergraph product ensures that this can be implemented in parallel. The boundary maps have the Kronecker form $A\otimes I$ and $I\otimes B$, where $A$ and $B$ are the base matrices: the factor carrying the ring elements acts on one base index, while the identity acts on the other. Because of the identity factor, the same ring element---and hence the same lift shift---is replicated across all blocks of the untouched factor, so a single global move implements that part of the layer across all of them in parallel, with the two AOD axes playing the roles of the two tensor factors. This is the mechanism by which the block structure already repeats in a product fashion in Ref.~\cite{xu2024constant}; for a cyclic lift it additionally makes the \emph{shifts} repeat, so they parallelize as well.

The same applies to the transitions between successive CNOT layers, which are the moves actually executed. A transition is a product of consecutive ring monomials, and because these all lie in the abelian group $\mathbb{Z}_\ell$, the product is again a single monomial---one uniform cyclic shift. Every layer, and every inter-layer transition, is therefore a single rigid shift applied uniformly to an entire block; for lifted product codes the move is usually just this uniform block shift.

Finally, the CRT layout of \cref{app:crt} controls the cost along each axis exactly as before. When the cyclic factors have coprime orders---or when a single composite lift $\mathbb{Z}_{ml}$ with $\gcd(m,l)=1$ is folded onto an $m\times l$ grid---the translation separates into independent row and column AOD shifts with no carry correction~\cite{wang2026coprime}.

\begin{figure*}[t]
  \centering
  \includegraphics[width=2\columnwidth]{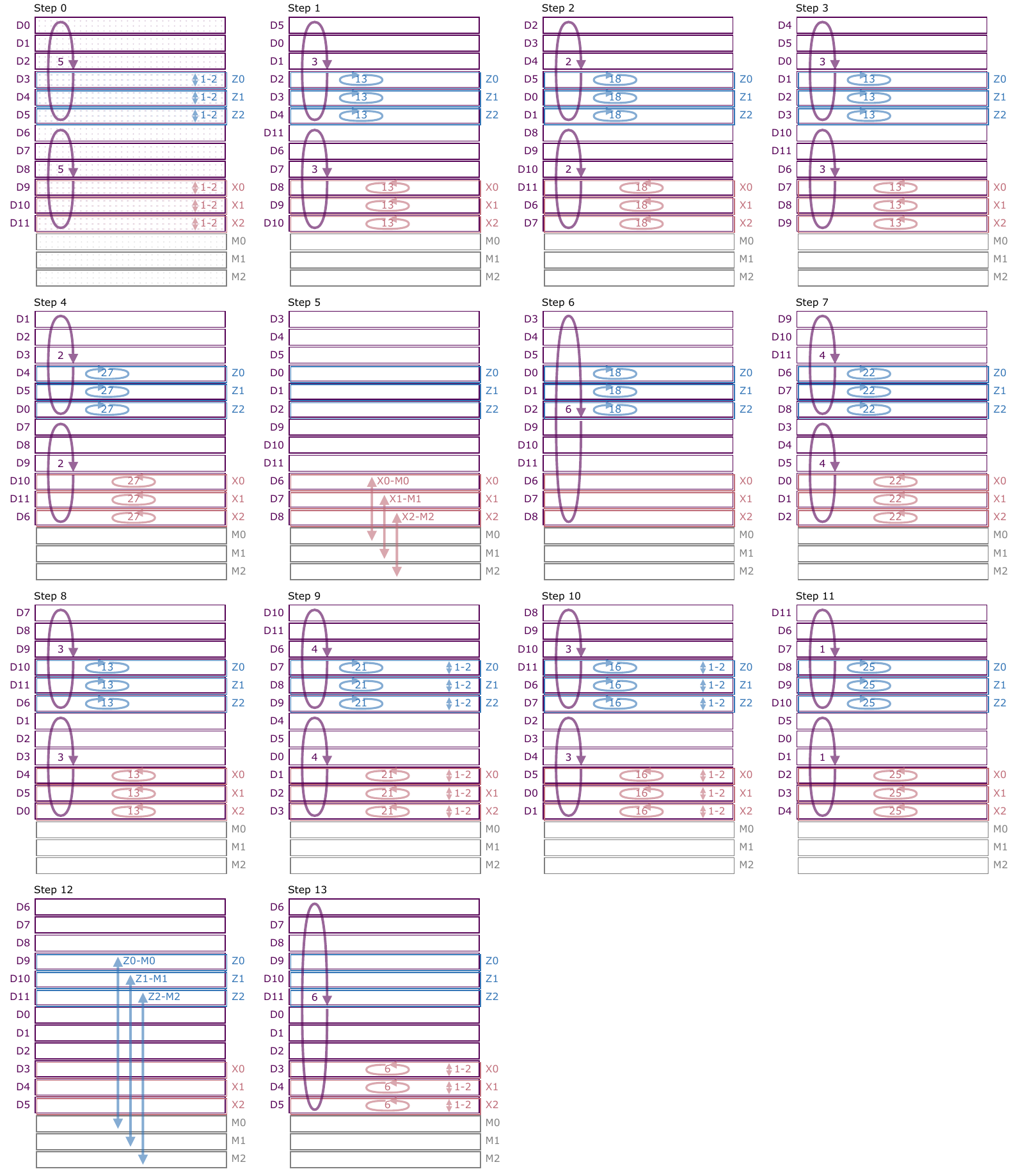}
  \caption{Atom positions and rearrangement sequence for the \nkdninetysix{} code ($P=96$). Positions are drawn to scale: data qubits are spaced $12\,\mu\mathrm{m}$ apart both vertically and horizontally, and each ancilla qubit sits $2\,\mu\mathrm{m}$ to the right of its associated data qubit. Arrows indicate atom rearrangements between or within blocks. Vertical circular arrows denote cyclic permutations of data blocks in which the bottom $n$ blocks are moved to the top. Horizontal circular arrows denote cyclic permutations of check qubits in which the rightmost $n$ columns are moved to the left. Double-ended arrows denote swaps between specific atom rows or blocks as labeled. CZ gates (not shown) are applied between each rearrangement step as well as before Step~0, but not between Steps~5 and~6 or between Steps~12 and~13.}
  \label{fig:P96-full-move-sequence}
\end{figure*}

\begin{figure*}[t]
  \centering
  \includegraphics[width=2\columnwidth]{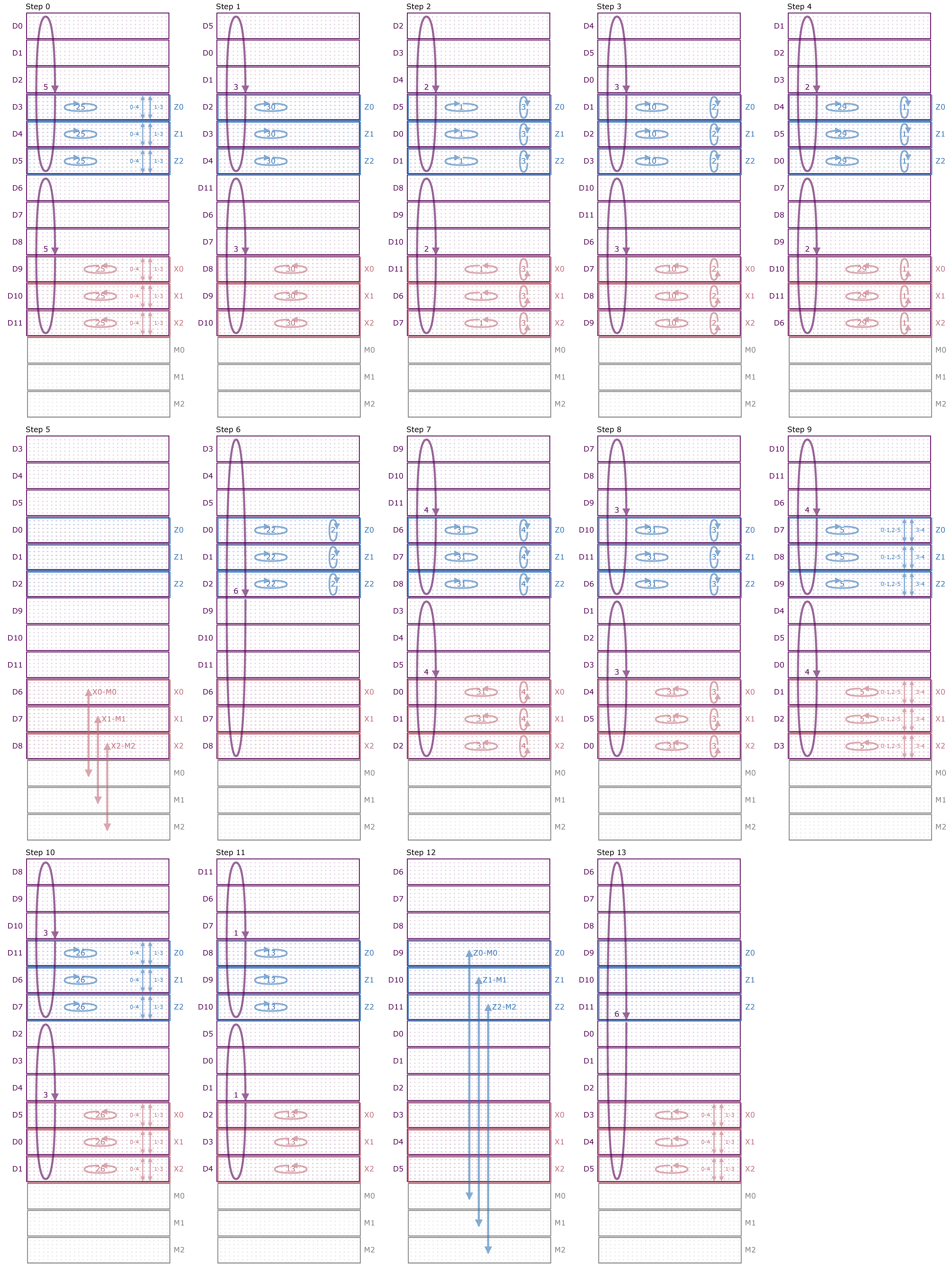}
  \caption{Atom positions and rearrangement sequence for the \nkdoneninetwo{} code ($P=192$). The layout and notation follow those of Fig.~\ref{fig:P96-full-move-sequence}. }
  \label{fig:P192-full-move-sequence}
\end{figure*}

%% file: main.bib
@article{tomamichel2016quantum,archivePrefix = {arXiv},arxivId = {1504.04617},author = {Tomamichel, Marco and Berta, Mario and Renes, Joseph M.},doi = {10.1038/ncomms11419},issn = {20411723},journal = {Nature Communications},title = {{Quantum coding with finite resources}},volume = {7},year = {2016}}

@article{polyanskiy2010channel,author = {Polyanskiy, Yury and Poor, H. Vincent and Verd{\'{u}}, Sergio},doi = {10.1109/TIT.2010.2043769},issn = {00189448},journal = {IEEE Transactions on Information Theory},number = {5},pages = {2307--2359},title = {{Channel coding rate in the finite blocklength regime}},volume = {56},year = {2010}}

@article{kuo2025generalized,archivePrefix = {arXiv},arxivId = {2310.12682v2},author = {Kuo, Kao Yueh and Lai, Ching Yi},doi = {10.1109/TIT.2025.3529773},issn = {15579654},journal = {IEEE Transactions on Information Theory},number = {3},pages = {1824--1840},publisher = {Institute of Electrical and Electronics Engineers Inc.},title = {{Generalized Quantum Data-Syndrome Codes and Belief Propagation Decoding for Phenomenological Noise}},url = {http://arxiv.org/abs/2310.12682 http://dx.doi.org/10.1109/TIT.2025.3529773},volume = {71},year = {2025}}

@article{skoric2023parallel,archivePrefix = {arXiv},arxivId = {2209.08552},author = {Skoric, Luka and Browne, Dan E. and Barnes, Kenton M. and Gillespie, Neil I. and Campbell, Earl T.},doi = {10.1038/s41467-023-42482-1},issn = {2041-1723},journal = {Nature Communications},number = {1},pages = {7040},title = {{Parallel window decoding enables scalable fault tolerant quantum computation}},url = {https://arxiv.org/abs/2209.08552v1 http://arxiv.org/abs/2209.08552 https://www.nature.com/articles/s41467-023-42482-1},volume = {14},year = {2023}}

@article{daguerre2024code,archivePrefix = {arXiv},arxivId = {2410.07327},author = {Daguerre, Lucas and Kim, Isaac H.},title = {{Code switching revisited: low-overhead magic state preparation using color codes}},url = {https://arxiv.org/abs/2410.07327v2},year = {2024},journal = {arXiv preprint arXiv:2410.07327}}

@article{muller2025improved,archivePrefix = {arXiv},arxivId = {2506.01779},author = {M{\"{u}}ller, Tristan and Alexander, Thomas and Beverland, Michael E. and B{\"{u}}hler, Markus and Johnson, Blake R. and Maurer, Thilo and Vandeth, Drew},title = {{Improved belief propagation is sufficient for real-time decoding of quantum memory}},url = {https://arxiv.org/pdf/2506.01779},year = {2025},journal = {arXiv preprint arXiv:2506.01779}}

@article{tillich2014quantum,author = {Tillich, Jean Pierre and Zemor, Gilles},doi = {10.1109/TIT.2013.2292061},journal = {IEEE Transactions on Information Theory},number = {2},pages = {1193--1202},publisher = {Institute of Electrical and Electronics Engineers Inc.},title = {{Quantum LDPC codes with positive rate and minimum distance proportional to the square root of the blocklength}},volume = {60},year = {2014}}

@article{fowler2012surface,author = {Fowler, Austin G. and Mariantoni, Matteo and Martinis, John M. and Cleland, Andrew N.},doi = {10.1103/PhysRevA.86.032324},issn = {1050-2947},journal = {Physical Review A},number = {3},pages = {032324},publisher = {American Physical Society},title = {{Surface codes: Towards practical large-scale quantum computation}},url = {https://link.aps.org/doi/10.1103/PhysRevA.86.032324 http://link.aps.org/doi/10.1103/PhysRevA.86.032324},volume = {86},year = {2012}}

@article{pryadko2022qdistrnd,author = {Pryadko, Leonid P and Shabashov, Vadim A and Kozin, Valerii K},doi = {10.21105/joss.04120},journal = {Journal of Open Source Software},number = {71},pages = {4120},publisher = {The Open Journal},title = {{QDistRnd: A GAP package for computing the distance of quantum error-correcting codes}},url = {https://doi.org/10.21105/joss.04120},volume = {7},year = {2022}}

@article{bravyi2013simulation,archivePrefix = {arXiv},arxivId = {1308.6270},author = {Bravyi, Sergey and Vargo, Alexander},doi = {10.1103/PHYSREVA.88.062308/FIGURES/11/MEDIUM},issn = {10502947},journal = {Physical Review A - Atomic, Molecular, and Optical Physics},number = {6},pages = {062308},publisher = {American Physical Society},title = {{Simulation of rare events in quantum error correction}},url = {https://journals.aps.org/pra/abstract/10.1103/PhysRevA.88.062308},volume = {88},year = {2013}}

@article{bonilla2025neural,archivePrefix = {arXiv},arxivId = {2509.11370},author = {{Bonilla Ataides}, J. Pablo and Gu, Andi and Yelin, Susanne F and Lukin, Mikhail D},title = {{Neural Decoders for Universal Quantum Algorithms}},url = {https://arxiv.org/pdf/2509.11370},year = {2025},journal = {arXiv preprint arXiv:2509.11370}}

@misc{gurobi,author = {{Gurobi Optimization, LLC}},title = {{Gurobi Optimizer Reference Manual}},url = {https://www.gurobi.com},year = {2024}}

@article{dennis2002topological,archivePrefix = {arXiv},arxivId = {quant-ph/0110143},author = {Dennis, Eric and Kitaev, Alexei and Landahl, Andrew and Preskill, John},doi = {10.1063/1.1499754},issn = {00222488},journal = {Journal of Mathematical Physics},number = {9},pages = {4452--4505},primaryClass = {quant-ph},title = {{Topological quantum memory}},volume = {43},year = {2002}}

@article{beni2025tesseract,archivePrefix = {arXiv},arxivId = {2503.10988},author = {Beni, Laleh Aghababaie and Higgott, Oscar and Shutty, Noah},title = {{Tesseract: A Search-Based Decoder for Quantum Error Correction}},url = {https://arxiv.org/pdf/2503.10988},year = {2025},journal = {arXiv preprint arXiv:2503.10988}}

@article{sahay2025fold,archivePrefix = {arXiv},arxivId = {2509.05212},author = {Sahay, Kaavya and Tsai, Pei-Kai and Chang, Kathleen and Su, Qile and Smith, Thomas B. and Singh, Shraddha and Puri, Shruti},title = {{Fold-transversal surface code cultivation}},url = {https://arxiv.org/pdf/2509.05212},year = {2025},journal = {arXiv preprint arXiv:2509.05212}}

@article{bausch2024learning,archivePrefix = {arXiv},arxivId = {2310.05900},author = {Bausch, Johannes and Senior, Andrew W and Heras, Francisco J H and Edlich, Thomas and Davies, Alex and Newman, Michael and Jones, Cody and Satzinger, Kevin and Niu, Murphy Yuezhen and Blackwell, Sam and Holland, George and Kafri, Dvir and Atalaya, Juan and Gidney, Craig and Hassabis, Demis and Boixo, Sergio and Neven, Hartmut and Kohli, Pushmeet},doi = {10.1038/s41586-024-08148-8},issn = {0028-0836},journal = {Nature},number = {8040},pages = {834--840},title = {{Learning high-accuracy error decoding for quantum processors}},url = {https://arxiv.org/abs/2310.05900v1 http://arxiv.org/abs/2310.05900%0Ahttp://dx.doi.org/10.1038/s41586-024-08148-8 https://www.nature.com/articles/s41586-024-08148-8},volume = {635},year = {2024}}

@article{breuckmann2017hyperbolic,archivePrefix = {arXiv},arxivId = {1703.00590v2},author = {Breuckmann, Nikolas P. and Vuillot, Christophe and Campbell, Earl and Krishna, Anirudh and Terhal, Barbara M.},doi = {10.1088/2058-9565/aa7d3b},journal = {Quantum Science and Technology},number = {3},publisher = {Institute of Physics Publishing},title = {{Hyperbolic and Semi-Hyperbolic Surface Codes for Quantum Storage}},url = {http://arxiv.org/abs/1703.00590 http://dx.doi.org/10.1088/2058-9565/aa7d3b},volume = {2},year = {2017}}

@article{xu2024constant,archivePrefix = {arXiv},arxivId = {2308.08648},author = {Xu, Qian and {Bonilla Ataides}, J. Pablo and Pattison, Christopher A and Raveendran, Nithin and Bluvstein, Dolev and Wurtz, Jonathan and Vasi{\'{c}}, Bane and Lukin, Mikhail D and Jiang, Liang and Zhou, Hengyun},doi = {10.1038/s41567-024-02479-z},issn = {17452481},journal = {Nature Physics},title = {{Constant-overhead fault-tolerant quantum computation with reconfigurable atom arrays}},url = {https://arxiv.org/abs/2308.08648v1 http://arxiv.org/abs/2308.08648 https://www.nature.com/articles/s41567-024-02479-z},year = {2024}}

@article{baspin2025fast,archivePrefix = {arXiv},arxivId = {2510.04521},author = {Baspin, Nou{\'{e}}dyn and Berent, Lucas and Cohen, Lawrence Z.},title = {{Fast surgery for quantum LDPC codes}},url = {https://arxiv.org/pdf/2510.04521},year = {2025},journal = {arXiv preprint arXiv:2510.04521}}

@article{blue2025machine,archivePrefix = {arXiv},arxivId = {2504.13043},author = {Blue, John and Avlani, Harshil and He, Zhiyang and Ziyin, Liu and Chuang, Isaac L.},title = {{Machine Learning Decoding of Circuit-Level Noise for Bivariate Bicycle Codes}},url = {https://arxiv.org/pdf/2504.13043},year = {2025},journal = {arXiv preprint arXiv:2504.13043}}

@article{yoder2025tour,archivePrefix = {arXiv},arxivId = {2506.03094},author = {Yoder, Theodore J. and Schoute, Eddie and Rall, Patrick and Pritchett, Emily and Gambetta, Jay M. and Cross, Andrew W. and Carroll, Malcolm and Beverland, Michael E.},title = {{Tour de gross: A modular quantum computer based on bivariate bicycle codes}},url = {https://arxiv.org/pdf/2506.03094},year = {2025},journal = {arXiv preprint arXiv:2506.03094}}

@book{richardson2008modern,author = {Richardson, Tom and Urbanke, Ruediger},publisher = {Cambridge University Press},title = {{Modern coding theory}},year = {2008}}

@article{landahl2011fault,archivePrefix = {arXiv},arxivId = {1108.5738},author = {Landahl, Andrew J. and Anderson, Jonas T. and Rice, Patrick R.},title = {{Fault-tolerant quantum computing with color codes}},url = {https://arxiv.org/abs/1108.5738v1 http://arxiv.org/abs/1108.5738},year = {2011},journal = {arXiv preprint arXiv:1108.5738}}

@article{kuo2022exploiting,archivePrefix = {arXiv},arxivId = {2104.13659},author = {Kuo, Kao Yueh and Lai, Ching Yi},doi = {10.1038/s41534-022-00623-2},issn = {2056-6387},journal = {npj Quantum Information 2022 8:1},number = {1},pages = {1--9},publisher = {Nature Publishing Group},title = {{Exploiting degeneracy in belief propagation decoding of quantum codes}},url = {https://www.nature.com/articles/s41534-022-00623-2},volume = {8},year = {2022}}

@article{lin2024quantum,archivePrefix = {arXiv},arxivId = {2306.16400},author = {Lin, Hsiang Ku and Pryadko, Leonid P.},doi = {10.1103/PhysRevA.109.022407},issn = {24699934},journal = {Physical Review A},number = {2},pages = {022407},publisher = {American Physical Society},title = {{Quantum two-block group algebra codes}},url = {https://journals.aps.org/pra/abstract/10.1103/PhysRevA.109.022407},volume = {109},year = {2024}}

@article{cain2026shor,archivePrefix = {arXiv},arxivId = {2603.28627},author = {Cain, Madelyn and Xu, Qian and King, Robbie and Picard, Lewis R. B. and Levine, Harry and Endres, Manuel and Preskill, John and Huang, Hsin-Yuan and Bluvstein, Dolev},title = {{Shor's algorithm is possible with as few as 10,000 reconfigurable atomic qubits}},url = {https://arxiv.org/pdf/2603.28627},year = {2026},journal = {arXiv preprint arXiv:2603.28627}}

@article{zhang2024time,archivePrefix = {arXiv},arxivId = {2408.01339v3},author = {Zhang, Guo and Li, Ying},doi = {10.1103/PhysRevLett.134.070602},journal = {Physical Review Letters},number = {7},publisher = {American Physical Society},title = {{Time-Efficient Logical Operations on Quantum Low-Density Parity Check Codes}},url = {http://arxiv.org/abs/2408.01339 http://dx.doi.org/10.1103/PhysRevLett.134.070602},volume = {134},year = {2024}}

@article{cohen2022low,archivePrefix = {arXiv},arxivId = {2110.10794},author = {Cohen, Lawrence Z. and Kim, Isaac H. and Bartlett, Stephen D. and Brown, Benjamin J.},doi = {10.1126/sciadv.abn1717},issn = {2375-2548},journal = {Science Advances},number = {20},title = {{Low-overhead fault-tolerant quantum computing using long-range connectivity}},url = {https://arxiv.org/abs/2110.10794v1 http://arxiv.org/abs/2110.10794 http://dx.doi.org/10.1126/sciadv.abn1717 https://www.science.org/doi/10.1126/sciadv.abn1717},volume = {8},year = {2022}}

@article{bravyi2024high,archivePrefix = {arXiv},arxivId = {2308.07915},author = {Bravyi, Sergey and Cross, Andrew W. and Gambetta, Jay M. and Maslov, Dmitri and Rall, Patrick and Yoder, Theodore J.},doi = {10.1038/s41586-024-07107-7},issn = {0028-0836},journal = {Nature},number = {8005},pages = {778--782},publisher = {Nature Publishing Group},title = {{High-threshold and low-overhead fault-tolerant quantum memory}},url = {https://www.nature.com/articles/s41586-024-07107-7 https://arxiv.org/abs/2308.07915v1},volume = {627},year = {2024}}

@article{tremblay2022constant,archivePrefix = {arXiv},arxivId = {2109.14609},author = {Tremblay, Maxime A. and Delfosse, Nicolas and Beverland, Michael E.},doi = {10.1103/PhysRevLett.129.050504},issn = {0031-9007},journal = {Physical Review Letters},number = {5},pages = {050504},title = {{Constant-Overhead Quantum Error Correction with Thin Planar Connectivity}},url = {https://arxiv.org/abs/2109.14609v1 http://arxiv.org/abs/2109.14609 http://dx.doi.org/10.1103/PhysRevLett.129.050504 https://link.aps.org/doi/10.1103/PhysRevLett.129.050504},volume = {129},year = {2022}}

@article{cross2024improved,archivePrefix = {arXiv},arxivId = {2407.18393},author = {Cross, Andrew and He, Zhiyang and Rall, Patrick and Yoder, Theodore},journal = {arXiv:2407.18393},title = {{Improved QLDPC Surgery: Logical Measurements and Bridging Codes}},url = {https://arxiv.org/pdf/2407.18393 http://arxiv.org/abs/2407.18393},year = {2024}}

@article{komoto2024quantum,archivePrefix = {arXiv},arxivId = {2412.21171},author = {Komoto, Daiki and Kasai, Kenta},title = {{Quantum Error Correction near the Coding Theoretical Bound}},url = {https://arxiv.org/abs/2412.21171v1},year = {2024},journal = {arXiv preprint arXiv:2412.21171}}

@article{ide2025fault,archivePrefix = {arXiv},arxivId = {2410.02753},author = {Ide, Benjamin and Gowda, Manoj G. and Nadkarni, Priya J. and Dauphinais, Guillaume},doi = {10.1103/PhysRevX.15.021088},issn = {2160-3308},journal = {Physical Review X},number = {2},pages = {021088},title = {{Fault-Tolerant Logical Measurements via Homological Measurement}},url = {https://arxiv.org/pdf/2410.02753 https://link.aps.org/doi/10.1103/PhysRevX.15.021088},volume = {15},year = {2025}}

@article{he2025extractors,archivePrefix = {arXiv},arxivId = {2503.10390},author = {He, Zhiyang and Cowtan, Alexander and Williamson, Dominic J. and Yoder, Theodore J.},isbn = {2503.10390v1},title = {{Extractors: QLDPC Architectures for Efficient Pauli-Based Computation}},url = {https://arxiv.org/pdf/2503.10390},year = {2025},journal = {arXiv preprint arXiv:2503.10390}}

@article{fahimniya2023fault,archivePrefix = {arXiv},arxivId = {2309.10033},author = {Fahimniya, Ali and Dehghani, Hossein and Bharti, Kishor and Mathew, Sheryl and Koll{\'{a}}r, Alicia J. and Gorshkov, Alexey V. and Gullans, Michael J.},title = {{Fault-tolerant hyperbolic Floquet quantum error correcting codes}},url = {https://arxiv.org/abs/2309.10033v2},year = {2023},journal = {arXiv preprint arXiv:2309.10033}}

@article{swaroop2024universal,archivePrefix = {arXiv},arxivId = {2410.03628},author = {Swaroop, Esha and Jochym-O'Connor, Tomas and Yoder, Theodore J.},title = {{Universal adapters between quantum LDPC codes}},url = {https://arxiv.org/pdf/2410.03628},year = {2024},journal = {arXiv preprint arXiv:2410.03628}}

@inproceedings{zhou2025resource,address = {New York, NY, USA},author = {Zhou, Hengyun and Duckering, Casey and Zhao, Chen and Bluvstein, Dolev and Cain, Madelyn and Kubica, Aleksander and Wang, Sheng-Tao and Lukin, Mikhail D.},booktitle = {Proceedings of the 52nd Annual International Symposium on Computer Architecture},doi = {10.1145/3695053.3731039},isbn = {9798400712616},pages = {1432--1448},publisher = {ACM},title = {{Resource Analysis of Low-Overhead Transversal Architectures for Reconfigurable Atom Arrays}},url = {https://dl.acm.org/doi/10.1145/3695053.3731039},year = {2025}}

@article{battistel2023real,archivePrefix = {arXiv},arxivId = {2303.00054v2},author = {Battistel, Francesco and Chamberland, Christopher and Johar, Kauser and Overwater, Ramon W. J. and Sebastiano, Fabio and Skoric, Luka and Ueno, Yosuke and Usman, Muhammad},doi = {10.1088/2399-1984/aceba6},title = {{Real-Time Decoding for Fault-Tolerant Quantum Computing: Progress, Challenges and Outlook}},url = {http://arxiv.org/abs/2303.00054 http://dx.doi.org/10.1088/2399-1984/aceba6},year = {2023},journal = {arXiv preprint arXiv:2303.00054v2}}

@article{strikis2026high,archivePrefix = {arXiv},arxivId = {2603.05481},author = {Strikis, Armands and Browne, Dan E. and Beverland, Michael E.},title = {{High-performance syndrome extraction circuits for quantum codes}},url = {https://arxiv.org/pdf/2603.05481},year = {2026},journal = {arXiv preprint arXiv:2603.05481}}

@article{lee2021even,author = {Lee, Joonho and Berry, Dominic W. and Gidney, Craig and Huggins, William J. and McClean, Jarrod R. and Wiebe, Nathan and Babbush, Ryan},doi = {10.1103/PRXQuantum.2.030305},journal = {PRX Quantum},number = {3},pages = {030305},publisher = {American Physical Society},title = {{Even More Efficient Quantum Computations of Chemistry Through Tensor Hypercontraction}},url = {https://journals.aps.org/prxquantum/abstract/10.1103/PRXQuantum.2.030305},volume = {2},year = {2021}}

@article{poulin2008on,archivePrefix = {arXiv},arxivId = {0801.1241},author = {Poulin, David and Chung, Yeojin},doi = {10.48550/arxiv.0801.1241},issn = {15337146},journal = {Quantum Information and Computation},number = {10},pages = {0987--1000},publisher = {Rinton Press Inc.},title = {{On the iterative decoding of sparse quantum codes}},url = {https://arxiv.org/abs/0801.1241v2},volume = {8},year = {2008}}

@article{wang2026coprime,archivePrefix = {arXiv},arxivId = {2408.10001v6},author = {Wang, Ming and Mueller, Frank},doi = {10.22331/q-2026-02-23-2009},issn = {2521-327X},journal = {Quantum},pages = {2009},publisher = {Verein zur F{\"{o}}rderung des Open Access Publizierens in den Quantenwissenschaften},title = {{Coprime Bivariate Bicycle Codes and Their Layouts on Cold Atoms}},url = {https://quantum-journal.org/papers/q-2026-02-23-2009/},volume = {10},year = {2026}}

@article{gu2026scalable,archivePrefix = {arXiv},arxivId = {2604.08358},author = {Gu, Andi and Pablo, J and Ataides, Bonilla and Lukin, Mikhail D and Yelin, Susanne F},title = {{Scalable Neural Decoders for Practical Fault-Tolerant Quantum Computation}},url = {https://arxiv.org/pdf/2604.08358},year = {2026},journal = {arXiv preprint arXiv:2604.08358}}

@article{kasai2026breaking,archivePrefix = {arXiv},arxivId = {2601.08824},author = {Kasai, Kenta},title = {{Breaking the Orthogonality Barrier in Quantum LDPC Codes}},url = {https://arxiv.org/pdf/2601.08824},year = {2026},journal = {arXiv preprint arXiv:2601.08824}}

@article{higgott2024constructions,archivePrefix = {arXiv},arxivId = {2308.03750},author = {Higgott, Oscar and Breuckmann, Nikolas P.},doi = {10.1103/PRXQuantum.5.040327},issn = {2691-3399},journal = {PRX Quantum},number = {4},pages = {040327},title = {{Constructions and Performance of Hyperbolic and Semi-Hyperbolic Floquet Codes}},url = {https://arxiv.org/abs/2308.03750v1 https://link.aps.org/doi/10.1103/PRXQuantum.5.040327},volume = {5},year = {2024}}

@article{breuckmann2021quantum,archivePrefix = {arXiv},arxivId = {2103.06309},author = {Breuckmann, Nikolas P. and Eberhardt, Jens Niklas},doi = {10.1103/prxquantum.2.040101},journal = {PRX Quantum},number = {4},publisher = {American Physical Society},title = {{Quantum Low-Density Parity-Check Codes}},url = {https://arxiv.org/abs/2103.06309v2},volume = {2},year = {2021}}

@article{delfosse2020hierarchical,archivePrefix = {arXiv},arxivId = {2001.11427},author = {Delfosse, Nicolas},title = {{Hierarchical decoding to reduce hardware requirements for quantum computing}},url = {https://arxiv.org/pdf/2001.11427},year = {2020},journal = {arXiv preprint arXiv:2001.11427}}

@article{bluvstein2024logical,archivePrefix = {arXiv},arxivId = {2312.03982},author = {Bluvstein, Dolev and Evered, Simon J. and Geim, Alexandra A. and Li, Sophie H. and Zhou, Hengyun and Manovitz, Tom and Ebadi, Sepehr and Cain, Madelyn and Kalinowski, Marcin and Hangleiter, Dominik and {Bonilla Ataides}, J. Pablo and Maskara, Nishad and Cong, Iris and Gao, Xun and {Sales Rodriguez}, Pedro and Karolyshyn, Thomas and Semeghini, Giulia and Gullans, Michael J. and Greiner, Markus and Vuleti{\'{c}}, Vladan and Lukin, Mikhail D.},doi = {10.1038/s41586-023-06927-3},issn = {0028-0836},journal = {Nature},number = {7997},pages = {58--65},pmid = {38056497},publisher = {Nature Publishing Group},title = {{Logical quantum processor based on reconfigurable atom arrays}},url = {https://www.nature.com/articles/s41586-023-06927-3},volume = {626},year = {2024}}

@article{bennett1996mixed,archivePrefix = {arXiv},arxivId = {quant-ph/9604024},author = {Bennett, Charles H. and DiVincenzo, David P. and Smolin, John A. and Wootters, William K.},doi = {10.1103/PhysRevA.54.3824},issn = {10941622},journal = {Physical Review A},number = {5},pages = {3824},pmid = {9913930},primaryClass = {quant-ph},publisher = {American Physical Society},title = {{Mixed-state entanglement and quantum error correction}},url = {https://journals.aps.org/pra/abstract/10.1103/PhysRevA.54.3824},volume = {54},year = {1996}}

@article{cain2024correlated,archivePrefix = {arXiv},arxivId = {2403.03272},author = {Cain, Madelyn and Zhao, Chen and Zhou, Hengyun and Meister, Nadine and Ataides, J. Pablo Bonilla and Jaffe, Arthur and Bluvstein, Dolev and Lukin, Mikhail D},doi = {10.1103/PhysRevLett.133.240602},issn = {0031-9007},journal = {Physical Review Letters},number = {24},pages = {240602},title = {{Correlated Decoding of Logical Algorithms with Transversal Gates}},url = {https://arxiv.org/abs/2403.03272v1 http://arxiv.org/abs/2403.03272 https://link.aps.org/doi/10.1103/PhysRevLett.133.240602},volume = {133},year = {2024}}

@article{grbic2026accelerating,archivePrefix = {arXiv},arxivId = {2602.02985},author = {Grbic, Dragana and Beni, Laleh Aghababaie and Shutty, Noah},title = {{Accelerating the Tesseract Decoder for Quantum Error Correction}},url = {http://arxiv.org/abs/2602.02985},year = {2026},journal = {arXiv preprint arXiv:2602.02985}}

@article{das2020scalable,archivePrefix = {arXiv},arxivId = {2001.06598},author = {Das, Poulami and Pattison, Christopher A. and Manne, Srilatha and Carmean, Douglas and Svore, Krysta and Qureshi, Moinuddin and Delfosse, Nicolas},doi = {10.48550/arxiv.2001.06598},title = {{A Scalable Decoder Micro-architecture for Fault-Tolerant Quantum Computing}},url = {https://arxiv.org/abs/2001.06598v1},year = {2020},journal = {arXiv preprint arXiv:2001.06598}}

@article{norcia2024iterative,archivePrefix = {arXiv},arxivId = {2401.16177},author = {Norcia, M. A. and Kim, H. and Cairncross, W. B. and Stone, M. and Ryou, A. and Jaffe, M. and Brown, M. O. and Barnes, K. and Battaglino, P. and Bohdanowicz, T. C. and Brown, A. and Cassella, K. and Chen, C. A. and Coxe, R. and Crow, D. and Epstein, J. and Griger, C. and Halperin, E. and Hummel, F. and Jones, A. M.W. and Kindem, J. M. and King, J. and Kotru, K. and Lauigan, J. and Li, M. and Lu, M. and Megidish, E. and Marjanovic, J. and McDonald, M. and Mittiga, T. and Muniz, J. A. and Narayanaswami, S. and Nishiguchi, C. and Paule, T. and Pawlak, K. A. and Peng, L. S. and Pudenz, K. L. and {Rodr{\'{i}}guez P{\'{e}}rez}, D. and Smull, A. and Stack, D. and Urbanek, M. and {Van De Veerdonk}, R. J.M. and Vendeiro, Z. and Wadleigh, L. and Wilkason, T. and Wu, T. Y. and Xie, X. and Zalys-Geller, E. and Zhang, X. and Bloom, B. J.},doi = {10.1103/PRXQuantum.5.030316},issn = {26913399},journal = {PRX Quantum},number = {3},pages = {030316},publisher = {American Physical Society},title = {{Iterative Assembly of 171 Yb Atom Arrays with Cavity-Enhanced Optical Lattices}},url = {http://arxiv.org/abs/2401.16177 https://link.aps.org/doi/10.1103/PRXQuantum.5.030316},volume = {5},year = {2024}}

@article{manetsch2024tweezer,archivePrefix = {arXiv},arxivId = {2403.12021},author = {Manetsch, Hannah J. and Nomura, Gyohei and Bataille, Elie and Leung, Kon H. and Lv, Xudong and Endres, Manuel},title = {{A tweezer array with 6100 highly coherent atomic qubits}},url = {https://arxiv.org/abs/2403.12021v2},year = {2024},journal = {arXiv preprint arXiv:2403.12021}}

@article{reiher2017elucidating,archivePrefix = {arXiv},arxivId = {1605.03590},author = {Reiher, Markus and Wiebe, Nathan and Svore, Krysta M. and Wecker, Dave and Troyer, Matthias},doi = {10.1073/pnas.1619152114},issn = {0027-8424},journal = {Proceedings of the National Academy of Sciences},number = {29},pages = {7555--7560},pmid = {28674011},publisher = {National Academy of Sciences},title = {{Elucidating reaction mechanisms on quantum computers}},url = {https://www.pnas.org/doi/abs/10.1073/pnas.1619152114 https://pnas.org/doi/full/10.1073/pnas.1619152114},volume = {114},year = {2017}}

@article{maurer2025real,archivePrefix = {arXiv},arxivId = {2510.21600},author = {Maurer, Thilo and B{\"{u}}hler, Markus and Kr{\"{o}}ner, Michael and Haverkamp, Frank and M{\"{u}}ller, Tristan and Vandeth, Drew and Johnson, Blake R.},title = {{Real-time decoding of the gross code memory with FPGAs}},url = {https://arxiv.org/pdf/2510.21600},year = {2025},journal = {arXiv preprint arXiv:2510.21600}}

@article{pause2024supercharged,archivePrefix = {arXiv},arxivId = {2310.09191},author = {Pause, Lars and Sturm, Lukas and Mittenb{\"{u}}hler, Marcel and Amann, Stephan and Preuschoff, Tilman and Sch{\"{a}}ffner, Dominik and Schlosser, Malte and Birkl, Gerhard},doi = {10.1364/optica.513551},issn = {23342536},journal = {Optica},number = {2},pages = {222},publisher = {Optica Publishing Group},title = {{Supercharged two-dimensional tweezer array with more than 1000 atomic qubits}},url = {http://arxiv.org/abs/2310.09191 http://dx.doi.org/10.1364/OPTICA.513551},volume = {11},year = {2024}}

@article{jenkins2022ytterbium,archivePrefix = {arXiv},arxivId = {2112.06732},author = {Jenkins, Alec and Lis, Joanna W. and Senoo, Aruku and McGrew, William F. and Kaufman, Adam M.},doi = {10.1103/PhysRevX.12.021027},issn = {2160-3308},journal = {Physical Review X},number = {2},pages = {021027},title = {{Ytterbium Nuclear-Spin Qubits in an Optical Tweezer Array}},url = {https://arxiv.org/abs/2112.06732v1 https://link.aps.org/doi/10.1103/PhysRevX.12.021027},volume = {12},year = {2022}}

@article{gidney2019how,archivePrefix = {arXiv},arxivId = {1905.09749},author = {Gidney, Craig and Eker{\aa}, Martin},doi = {10.22331/q-2021-04-15-433},journal = {Quantum},pages = {1--31},publisher = {Verein zur Forderung des Open Access Publizierens in den Quantenwissenschaften},title = {{How to factor 2048 bit RSA integers in 8 hours using 20 million noisy qubits}},url = {https://arxiv.org/abs/1905.09749v3},volume = {5},year = {2019}}

@article{cicali2025fast,author = {Cicali, Cristina and Calzavara, Martino and Cuestas, Eloisa and Calarco, Tommaso and Zeier, Robert and Motzoi, Felix},doi = {10.1103/7r3w-8m61},issn = {23317019},journal = {Physical Review Applied},number = {2},pages = {024070},publisher = {American Physical Society},title = {{Fast neutral-atom transport and transfer between optical tweezers}},url = {https://journals.aps.org/prapplied/abstract/10.1103/7r3w-8m61},volume = {24},year = {2025}}

@article{hwang2025fast,archivePrefix = {arXiv},arxivId = {2410.22627},author = {Hwang, Sunhwa and Hwang, Hansub and Kim, Kangjin and Byun, Andrew and Kim, Kangheun and Jeong, Seokho and Soegianto, Maynardo Pratama and Soegianto, Maynardo Pratama and Ahn, Jaewook},doi = {10.1364/OPTICAQ.546797},issn = {2837-6714},journal = {Optica Quantum, Vol. 3, Issue 1, pp. 64-71},number = {1},pages = {64--71},publisher = {Optica Publishing Group},title = {{Fast and reliable atom transport by optical tweezers}},url = {https://opg.optica.org/viewmedia.cfm?uri=opticaq-3-1-64\&seq=0\&html=true https://opg.optica.org/abstract.cfm?uri=opticaq-3-1-64 https://opg.optica.org/opticaq/abstract.cfm?uri=opticaq-3-1-64},volume = {3},year = {2025}}

@article{scruby2024high,archivePrefix = {arXiv},arxivId = {2406.14445},author = {Scruby, Thomas R. and Hillmann, Timo and Roffe, Joschka},title = {{High-threshold, low-overhead and single-shot decodable fault-tolerant quantum memory}},url = {https://arxiv.org/pdf/2406.14445 http://arxiv.org/abs/2406.14445},year = {2024},journal = {arXiv preprint arXiv:2406.14445}}

@article{viszlai2023matching,archivePrefix = {arXiv},arxivId = {2311.16980},author = {Viszlai, Joshua and Yang, Willers and Lin, Sophia Fuhui and Liu, Junyu and Nottingham, Natalia and Baker, Jonathan M. and Chong, Frederic T.},title = {{Matching Generalized-Bicycle Codes to Neutral Atoms for Low-Overhead Fault-Tolerance}},url = {https://arxiv.org/abs/2311.16980v1},year = {2023},journal = {arXiv preprint arXiv:2311.16980}}

@article{breuckmann2020balanced,archivePrefix = {arXiv},arxivId = {2012.09271v3},author = {Breuckmann, Nikolas P. and Eberhardt, Jens N.},doi = {10.1109/TIT.2021.3097347},journal = {IEEE Transactions on Information Theory},number = {10},pages = {6653--6674},publisher = {Institute of Electrical and Electronics Engineers Inc.},title = {{Balanced Product Quantum Codes}},url = {http://arxiv.org/abs/2012.09271 http://dx.doi.org/10.1109/TIT.2021.3097347},volume = {67},year = {2020}}

@article{williamson2024low,archivePrefix = {arXiv},arxivId = {2410.02213},author = {Williamson, Dominic J. and Yoder, Theodore J.},title = {{Low-overhead fault-tolerant quantum computation by gauging logical operators}},url = {https://arxiv.org/abs/2410.02213v1 http://arxiv.org/abs/2410.02213},year = {2024},journal = {arXiv preprint arXiv:2410.02213}}

@techreport{etsi,author = {ETSI},institution = {European Telecommunications Standards Institute},month = {jul},number = {ETSI TS 138 212 V18.7.0},title = {{5G; NR; Multiplexing and channel coding (3GPP TS 38.212 version 18.7.0 Release 18)}},type = {ETSI Technical Specification},url = {https://www.etsi.org/deliver/etsi_ts/138200_138299/138212/18.07.00_60/ts_138212v180700p.pdf},year = {2025}}

@article{calderbank1996good,archivePrefix = {arXiv},arxivId = {quant-ph/9512032},author = {Calderbank, A. R. and Shor, Peter W.},doi = {10.1103/PhysRevA.54.1098},issn = {10941622},journal = {Physical Review A},number = {2},pages = {1098},pmid = {9913578},primaryClass = {quant-ph},publisher = {American Physical Society},title = {{Good quantum error-correcting codes exist}},url = {https://journals.aps.org/pra/abstract/10.1103/PhysRevA.54.1098},volume = {54},year = {1996}}

@article{steane1996error,author = {Steane, A. M.},doi = {10.1103/PhysRevLett.77.793},issn = {10797114},journal = {Physical Review Letters},number = {5},pages = {793},publisher = {American Physical Society},title = {{Error Correcting Codes in Quantum Theory}},url = {https://journals.aps.org/prl/abstract/10.1103/PhysRevLett.77.793},volume = {77},year = {1996}}

@article{gidney2024magic,archivePrefix = {arXiv},arxivId = {2409.17595},author = {Gidney, Craig and Shutty, Noah and Jones, Cody},title = {{Magic state cultivation: growing T states as cheap as CNOT gates}},url = {https://arxiv.org/abs/2409.17595v1},year = {2024},journal = {arXiv preprint arXiv:2409.17595}}

@article{kitaev2003fault,archivePrefix = {arXiv},arxivId = {quant-ph/9707021},author = {Kitaev, A. Yu},doi = {10.1016/S0003-4916(02)00018-0},issn = {00034916},journal = {Annals of Physics},number = {1},pages = {2--30},primaryClass = {quant-ph},publisher = {Academic Press},title = {{Fault-tolerant quantum computation by anyons}},volume = {303},year = {2003}}

@article{zhou2025opportunities,author = {Zhou, Hengyun and Cain, Madelyn and Lukin, Mikhail D.},doi = {10.1038/s43588-025-00895-6},issn = {2662-8457},journal = {Nature Computational Science 2025 5:12},number = {12},pages = {1110--1119},publisher = {Nature Publishing Group},title = {{Opportunities in full-stack design of low-overhead fault-tolerant quantum computation}},url = {https://www.nature.com/articles/s43588-025-00895-6},volume = {5},year = {2025}}

@online{ibm2023condor,author = {Gambetta, Jay},institution = {IBM Quantum},month = {dec},title = {{The Hardware and Software for the Era of Quantum Utility Is Here}},url = {https://www.ibm.com/quantum/blog/quantum-roadmap-2033},year = {2023}}

@article{panteleev2019degenerate,author = {Panteleev, Pavel and Kalachev, Gleb},doi = {10.22331/q-2021-11-22-585},issn = {2521327X},journal = {Quantum},pages = {585},title = {{Degenerate Quantum LDPC Codes With Good Finite Length Performance}},url = {https://arxiv.org/abs/1904.02703v3 https://quantum-journal.org/papers/q-2021-11-22-585/},volume = {5},year = {2019}}

@article{webster2026pinnacle,archivePrefix = {arXiv},arxivId = {2602.11457},author = {Webster, Paul and Berent, Lucas and Chandra, Omprakash and Hockings, Evan T. and Baspin, Nou{\'{e}}dyn and Thomsen, Felix and Smith, Samuel C. and Cohen, Lawrence Z.},title = {{The Pinnacle Architecture: Reducing the cost of breaking RSA-2048 to 100 000 physical qubits using quantum LDPC codes}},url = {https://arxiv.org/pdf/2602.11457},year = {2026},journal = {arXiv preprint arXiv:2602.11457}}

@article{bluvstein2022quantum,archivePrefix = {arXiv},arxivId = {2112.03923},author = {Bluvstein, Dolev and Levine, Harry and Semeghini, Giulia and Wang, Tout T. and Ebadi, Sepehr and Kalinowski, Marcin and Keesling, Alexander and Maskara, Nishad and Pichler, Hannes and Greiner, Markus and Vuleti{\'{c}}, Vladan and Lukin, Mikhail D.},doi = {10.1038/s41586-022-04592-6},issn = {14764687},journal = {Nature},number = {7906},pages = {451--456},pmid = {35444318},publisher = {Nature Publishing Group},title = {{A quantum processor based on coherent transport of entangled atom arrays}},url = {https://www.nature.com/articles/s41586-022-04592-6},volume = {604},year = {2022}}

@article{zhang2025leveraging,archivePrefix = {arXiv},arxivId = {2506.13724},author = {Zhang, Bichen and Liu, Genyue and Bornet, Guillaume and Horvath, Sebastian P. and Peng, Pai and Ma, Shuo and Huang, Shilin and Puri, Shruti and Thompson, Jeff D.},title = {{Leveraging erasure errors in logical qubits with metastable $^{171}$Yb atoms}},url = {https://arxiv.org/pdf/2506.13724},year = {2025},journal = {arXiv preprint arXiv:2506.13724}}

@article{zheng2025high,archivePrefix = {arXiv},arxivId = {2510.08523},author = {Zheng, Guo and Jiang, Liang and Xu, Qian},title = {{High-Rate Surgery: towards constant-overhead logical operations}},url = {https://arxiv.org/pdf/2510.08523},year = {2025},journal = {arXiv preprint arXiv:2510.08523}}

@article{cowtan2025fast,archivePrefix = {arXiv},arxivId = {2510.14895},author = {Cowtan, Alexander and He, Zhiyang and Williamson, Dominic J. and Yoder, Theodore J.},title = {{Fast and fault-tolerant logical measurements: Auxiliary hypergraphs and transversal surgery}},url = {https://arxiv.org/pdf/2510.14895},year = {2025},journal = {arXiv preprint arXiv:2510.14895}}

@article{chang2026constant,archivePrefix = {arXiv},arxivId = {2603.02157},author = {Chang, Kathleen and He, Zhiyang and Yoder, Theodore J. and Zhu, Guanyu and Jochym-O'Connor, Tomas},title = {{Constant-Time Surgery on 2D Hypergraph Product Codes with Near-Constant Space Overhead}},url = {https://arxiv.org/pdf/2603.02157},year = {2026},journal = {arXiv preprint arXiv:2603.02157}}

@article{xu2025batched,archivePrefix = {arXiv},arxivId = {2510.06159},author = {Xu, Qian and Zhou, Hengyun and Bluvstein, Dolev and Cain, Madelyn and Kalinowski, Marcin and Preskill, John and Lukin, Mikhail D. and Maskara, Nishad},title = {{Batched high-rate logical operations for quantum LDPC codes}},url = {https://arxiv.org/pdf/2510.06159},year = {2025},journal = {arXiv preprint arXiv:2510.06159}}

@article{beverland2025fail,archivePrefix = {arXiv},arxivId = {2511.15177},author = {Beverland, Michael E. and Carroll, Malcolm and Cross, Andrew W. and Yoder, Theodore J.},title = {{Fail fast: techniques to probe rare events in quantum error correction}},url = {https://arxiv.org/pdf/2511.15177},year = {2025},journal = {arXiv preprint arXiv:2511.15177}}

@article{menon2025magic,archivePrefix = {arXiv},arxivId = {2508.10714},author = {Menon, Varun and Bonilla-Ataides, J. Pablo and Mehta, Rohan and Gu, Andi and Tan, Daniel Bochen and Lukin, Mikhail D.},title = {{Magic tricycles: Efficient magic state generation with finite block-length quantum LDPC codes}},url = {https://arxiv.org/pdf/2508.10714},year = {2025},journal = {arXiv preprint arXiv:2508.10714}}

@inproceedings{panteleev2022asymptotically,archivePrefix = {arXiv},arxivId = {2111.03654},author = {Panteleev, Pavel and Kalachev, Gleb},booktitle = {Proceedings of the Annual ACM Symposium on Theory of Computing},doi = {10.1145/3519935.3520017},isbn = {9781450392648},issn = {07378017},pages = {375--388},publisher = {Association for Computing Machinery},title = {{Asymptotically good Quantum and locally testable classical LDPC codes}},url = {https://dl.acm.org/doi/10.1145/3519935.3520017},year = {2022},journal = {arXiv preprint arXiv:2111.03654}}

@article{sales2025experimental,archivePrefix = {arXiv},arxivId = {2412.15165},author = {{Sales Rodriguez}, Pedro and Robinson, John M and Jepsen, Paul Niklas and He, Zhiyang and Duckering, Casey and Zhao, Chen and Wu, Kai Hsin and Campo, Joseph and Bagnall, Kevin and Kwon, Minho and Karolyshyn, Thomas and Weinberg, Phillip and Cain, Madelyn and Evered, Simon J and Geim, Alexandra A and Kalinowski, Marcin and Li, Sophie H and Manovitz, Tom and Amato-Grill, Jesse and Basham, James I and Bernstein, Liane and Braverman, Boris and Bylinskii, Alexei and Choukri, Adam and DeAngelo, Robert J. and Fang, Fang and Fieweger, Connor and Frederick, Paige and Haines, David and Hamdan, Majd and Hammett, Julian and Hsu, Ning and Hu, Ming Guang and Huber, Florian and Jia, Ningyuan and Kedar, Dhruv and Kornja{\v{c}}a, Milan and Liu, Fangli and Long, John and Lopatin, Jonathan and Lopes, Pedro L.S. and Luo, Xiu Zhe and Macr{\`{i}}, Tommaso and Markovi{\'{c}}, Ognjen and Mart{\'{i}}nez-Mart{\'{i}}nez, Luis A and Meng, Xianmei and Ostermann, Stefan and Ostroumov, Evgeny and Paquette, David and Qiang, Zexuan and Shofman, Vadim and Singh, Anshuman and Singh, Manuj and Sinha, Nandan and Thoreen, Henry and Wan, Noel and Wang, Yiping and Waxman-Lenz, Daniel and Wong, Tak and Wurtz, Jonathan and Zhdanov, Andrii and Zheng, Laurent and Greiner, Markus and Keesling, Alexander and Gemelke, Nathan and Vuleti{\'{c}}, Vladan and Kitagawa, Takuya and Wang, Sheng Tao and Bluvstein, Dolev and Lukin, Mikhail D. and Lukin, Alexander and Zhou, Hengyun and Cant{\'{u}}, Sergio H},doi = {10.1038/s41586-025-09367-3},issn = {14764687},journal = {Nature},number = {8081},pages = {620--625},pmid = {40659049},title = {{Experimental demonstration of logical magic state distillation}},url = {https://arxiv.org/abs/2412.15165v1 http://arxiv.org/abs/2412.15165},volume = {645},year = {2025}}

@article{graham2022multi1,archivePrefix = {arXiv},arxivId = {2112.14589},author = {Graham, T. M. and Song, Y. and Scott, J. and Poole, C. and Phuttitarn, L. and Jooya, K. and Eichler, P. and Jiang, X. and Marra, A. and Grinkemeyer, B. and Kwon, M. and Ebert, M. and Cherek, J. and Lichtman, M. T. and Gillette, M. and Gilbert, J. and Bowman, D. and Ballance, T. and Campbell, C. and Dahl, E. D. and Crawford, O. and Blunt, N. S. and Rogers, B. and Noel, T. and Saffman, M.},doi = {10.1038/s41586-022-04603-6},issn = {0028-0836},journal = {Nature},number = {7906},pages = {457--462},pmid = {35444321},title = {{Multi-qubit entanglement and algorithms on a neutral-atom quantum computer}},url = {https://arxiv.org/abs/2112.14589v3 http://arxiv.org/abs/2112.14589 http://dx.doi.org/10.1038/s41586-022-04603-6 https://www.nature.com/articles/s41586-022-04603-6},volume = {604},year = {2022}}
